%% file: QLCM.tex
\documentclass[10pt,a4paper]{article}
\usepackage{a4wide}

\newcommand{\typeof}{0}   
 
\usepackage{ifthen}
\newcommand{\condinc}[2]
{
  \ifthenelse{\equal{\typeof}{0}}
  {#1}
  {#2}
}

\usepackage{pdfsync}

\input{global} 

\input{local}

\newcommand{\qstar}{\textsf{Q}\mbox{$^{*}$}}

\condinc{}{}

\begin{document}
\condinc{
\title{Confluence Results\\ for a Quantum Lambda Calculus\\ with Measurements}
\author{Ugo Dal Lago\footnote{Dipartimento di Scienze dell'Informazione, Universit\`a di Bologna} 
\and Andrea Masini\footnote{Dipartimento di Informatica, Universit\`a di Verona}
\and Margherita Zorzi\footnote{Dipartimento di Informatica, Universit\`a di Verona}}
\maketitle
\begin{abstract} 
  A strong confluence result for \qstar, a quantum $\lambda$-calculus with measurements, is proved.
  More precisely, confluence is shown to hold both for finite and infinite computations. The technique used in the
  confluence proof is syntactical but innovative. This makes \qstar\ different from 
  similar quantum lambda calculi, which are either measurement-free or provided with
  a reduction strategy.
\end{abstract}
}{%
\begin{frontmatter}
  \title{Confluence Results\\ for a Quantum Lambda Calculus\\ with Measurements\thanksref{ALL}} 
  \author{Ugo Dal Lago\thanksref{ugoemail}}
  \address{Dipartimento di Scienze dell'Informazione\\ Universit\`a di Bologna} 
  \author{Andrea Masini\thanksref{andreaemail}}
  \author{Margherita Zorzi\thanksref{margheemail}}
  \address{Dipartimento di Informatica\\Universit\`a di Verona}
  \thanks[ALL]{The authors are partially supported by PRIN project ``CONCERTO'' and FIRB grant RBIN04M8S8, ``Intern. Inst. for Applicable Math.''} 
  \thanks[ugoemail]{Email:\href{mailto:dallago@cs.unibo.it}
    {\texttt{\normalshape dallago@cs.unibo.it}}}
  \thanks[andreaemail]{Email:\href{mailto:andrea.masini@univr.it}
    {\texttt{\normalshape andrea.masini@univr.it}}}
  \thanks[margheemail]{Email:\href{mailto:margherita.zorzi@univr.it}
    {\texttt{\normalshape margherita.zorzi@univr.it}}}
\begin{abstract} 
  A strong confluence result for \qstar, a quantum $\lambda$-calculus with measurements, is proved.
  More precisely, confluence is shown to hold both for finite and infinite computations. The technique used in the
  confluence proof is syntactical but innovative. This makes \qstar\ different from 
  similar quantum lambda calculi, which are either measurement-free or provided with
  a reduction strategy.
\end{abstract}
\end{frontmatter}
}%
\section{Introduction}
It is well known that the measurement-free evolution of a quantum system is deterministic.  
As a consequence it is to be expected that a good measurement-free quantum lambda calculus 
enjoys confluence. This is the case of \qcalc, by the authors~\cite{DLMZmscs08} and
of the lambda calculus recently introduced by Arrighi and Dowek~\cite{ArrDow08RTA}.
The situation becomes more complicated if we introduce a measurement operator. In fact measurements 
break the deterministic evolution of a quantum system\footnote{at the present theoretical status of
quantum mechanics, measurement is not reconcilable with the quantum theory, this is the famous 
\textit{measurement problem}~\cite{Ish95}}: in presence of measurements the behaviour becomes irremediably probabilistic.

An explicit measurement operator in the syntax allows an observation at an intermediate step 
of the computation: this feature is needed if we want, for example, to write algorithms such 
as Shor's factorization. In quantum calculi the intended  meaning of a measurement is the observation of
a (possibly superimposed) quantum bit, giving as output a classical bit; the two possible
outcomes (i.e., the two possible values of the obtained classical bit) can be observed with two probabilities
summing to $1$. Since measurement forces a probabilistic evolution in the computation, it is not 
surprising that we need probabilistic instruments in order to investigate the main features of 
the language. 

In this paper, we study an extension of \qcalc\ obtained by endowing the
language of terms with a suitable measurement operator and coherently
extending the reduction relation, which becomes probabilistic for the
reasons we have just explained. 
We investigate the resulting calculus, called \qstar, focusing,
in particular, on confluence.

In \qstar\ and \qcalc, states are formalized by \emph{configurations},
i.e., triples in the form $[\Qr, \QV, M]$, where $M$ is a lambda term, $\Qr$ is
a quantum state, and $\QV$ is a set of names of quantum variables. 
So, control is classical ($M$ is simply a term) while data is
quantum ($\Qr$ is an element of a finite-dimensional
Hilbert space).

We are interested in the following question:
what happens to properties such as confluence in
presence of measurements? And moreover: is it possible to preserve
confluence in the probabilistic setting induced by 
measurements?
Apparently, the questions above cannot receive a positive answer:
as we will see in section~\ref{sec:ConfPb}, it is possible to exhibit a
configuration $\confone$ such that there are two
different reductions starting at $\confone$ and
ending in two essentially different configurations 
in normal form $[\unoqr,\emptyset,0]$ and
$[\unoqr,\emptyset,1]$. In other words, confluence fails in its
usual form.
But the question now becomes: are the usual notions of computations 
and confluence adequate in this setting?

In \qstar, there are two distinct sources of divergence:
\begin{varitemize}
\item
  On the one hand, a redex involving the measurement operator
  can be reduced in two different ways, i.e., divergence
  can come from \emph{a single redex}.
\item
  On the other hand, a term can contain
  more than one redex and \qstar\ is \emph{not} endowed with a
  reduction strategy. As a consequence, some configurations can
  be reduced in different ways due to the presence of
  \emph{distinct redexes} in a term.
\end{varitemize}
We cannot hope to be confluent with respect to the first source
of divergence, but we can anyway ask ourselves whether all reduction
strategies are somehow equivalent. More precisely, we say that 
\qstar\ is \emph{confluent} if
for every configuration $\confone$ and for every configuration in normal form 
$\conftwo$, there is a fixed real number $p$ such that the probability of observing 
$\conftwo$ when reducing $\confone$ is always $p$, independently of the reduction strategy.

This notion of confluence can be easily captured by analyzing 
rewriting on \emph{mixed states} rather than rewriting on
configurations. A mixed state is a probabilistic distribution
on configurations whose support is finite. Rewriting on
configurations naturally extend to rewriting on mixed states.
Rewriting on mixed states is \emph{not} a probabilistic relation,
and confluence is the usual confluence coming from rewriting
theory~\cite{TereseTRS}.

In this paper, we prove that \qstar\ is indeed confluent
in this sense. Technically, confluence is proved in an innovative way.
The key point is that we need a new definition of computation. The
usual notion of computation as a sequence of configurations
is not adequate here.
A notion of \emph{probabilistic computation} replaces it, 
as something more general than a linear sequence of configurations
but less general than the reduction tree:
a probabilistic computation is a (possibly) infinite tree,
in which binary choice (a node can have at most two children)
corresponds to the two possible outcomes of a measurement. 
This new notion of computation is needed, because proving 
confluence \emph{directly} on mixed states is non-trivial.
As by-products, we prove other results in the style of confluence.

Another important property of any quantum lambda calculus with
measurements is the importance of infinite computations. In the
case of standard lambda calculus, the study of infinite computations is
strongly related to the study of infinite lambda terms. This is not
the case of \qstar\  (and in general of quantum calculi with
measurements).  
This phenomenon forced us to extend the
study of confluence to the case of infinite probabilistic
computations.  The proposed analysis is not standard and is based
on new techniques.

The rest of this paper is structured as follow:
\begin{varitemize} 
\item in Section~\ref{sec:QStarWfr} the quantum $\lambda$-calculus \qstar\ is introduced;
\item in Section~\ref{sec:ConfPb} we introduce the confluence problem in an informal way;
\item in Section~\ref{sec:Pcomp} we give the definition of a \emph{probabilistic computation};
\item in Section~\ref{sec:StrongConfl} a strong confluence result on probabilistic computations is given;
\item in Section~\ref{sec:MixS} \emph{ mixed states} and  \emph{mixed computations} are introduced, and we give 
  a conluence theorem for mixed computations. 
\end{varitemize}
\condinc{}{An extended version with all proofs is available~\cite{ExtendedVersion}.}

\condinc{%
\section{The Calculus  \qstar}\label{sec:QStarWfr}
 In~\cite{DLMZmscs08} we have introduced a measurement--free, untyped quantum
$\lambda$--calculus, called \qcalc, based on the 
\textit{quantum data and classical control} paradigm (see
e.g.~\cite{Sel04c,SelVal06}). In this paper we generalize \qcalc\ by extending the class of terms
with a measurement operator, obtaining \qstar. 

As for \qcalc, \qstar\ is based on the notion of a \textit{configuration} (see Section~\ref{sect:comp}), 
namely a triple $[\Qr, \QV, M]$ where $\Qr$ is a \textit{quantum register}\footnote{the ``empty'' 
quantum register will be denoted with the scalar number $1$.}, $\QV$ is a \textit{finite} set of names, 
called \textit{quantum variables}, and $M$ is an untyped \textit{term} based on the linear lambda-calculus 
defined by Wadler~\cite{Wad94} and Simpson~\cite{Simpson05}.  

Quantum registers are systems of $n$ qubits, that, mathematically speaking, are normalized vectors of finite 
dimensional Hilbert spaces. In particular, a quantum register $\Qr$ of a configuration $[\Qr, \QV, M]$, is 
a normalized vector of the Hilbert space $\ell^2(\{0,1\}^\QV)$, denoted here with $\Hs{\QV}.$\footnote{see ~\cite{DLMZmscs08} 
for a full discussion of $\Hs{\QV}$ and ~\cite{RomanBook} for a general treatment of $\ell^2(S)$ spaces.}  
Roughly speaking, the reader could think that quantum variables are references 
to qubits in a quantum register. 

There are three kinds of operations on quantum registers: \textit{(i)} the 
creation of a new qubit; \textit{(ii)} \textit{unitary operators}: each unitary operator $\mathbf{U}_{<< q_1,\ldots,q_n>>}$ 
corresponds to a pure quantum operation acting on qubits with names $q_1,\ldots,q_n$ (mathematically, a unitary 
transform on the Hilbert space $\Hs{\{q_1,\ldots,q_n\}}$, see~\cite{DLMZmscs08}); \textit{(iii)} \textit{one qubit measurement} 
operations $\mis{r}{0}, \mis{r}{1}$ responsible of the probabilistic reduction of the quantum state plus 
the destruction of the qubit referenced by $r$: given a quantum register $\Qr\in\Hs{\QV}$, and a quantum variable 
name $r\in \QV$, we allow the (destructive) measurement of the qubit with name $r$\condinc{(see Section~\ref{sect:qrm} for more details).}
{ 
  \footnote{More precisely, for every quantum variable $r$ we assume
    the existence of two linear measurement operators, $\mis{r}{0}, \mis{r}{1} :
    \Hs{\QV} \to \Hs{\QV-\{r\}}$ enjoying the completeness condition
    ${\mis{r}{0}}^\dag\mis{r}{0} + {\mis{r}{1}}^\dag \mis{r}{1}=
    \mathit{id}_{\Hs{\QV}}$ and such that, given a quantum register
    $\Qr\in\Hs{\QV}$, the measurement of the qubit with name $r$ in
    $\Qr$ gives the outcome $c$ (with $c\in\{0,1\}$) with probability
    $p_c=\langle\Qr|{\mis{r}{c}}^\dag\mis{r}{c}|\Qr\rangle$ and produces
    the new quantum register $\frac{\mis{r}{c}\Qr}{\sqrt p_c}$;
    see~\cite{NieCh00,KLM07} for a detailed discussion of general pure
    measurements.}.
}

We conclude this short overview with few words on the set of \textit{elementary unitary operators}.
We say that a class $\{ \mathbf{U}_i\}_{i\in I}$ is \textit{elementary} iff for every $i\in I$ , 
the unitary operator $\mathbf{U}_i$ is realizable, either physically (i.e. by a laser or by other apparatus) 
or by means of a computable devices, such as a Turing machine. 
Different classes of elementary operators could be defined, among these, a remarkable class is those of 
\textit{computable operators}, see e.g.~\cite{BerVa97,DLMZmscs08,NieCh00,NiOz02,NishOza08}.
  
\subsection{Terms, Judgements and Well-Formed-Terms}

Let ${\mathcal U}$ an elementary set of unitary operators. Let us associate to each elementary 
operator $\mathbf{U}\in {\mathcal U}$ a symbol $U$.
The set of \emph{term expressions}, or
\emph{terms} for short, is defined by the following grammar:
$$\begin{array}{lcl}
  x & ::= & x_0, x_1,\ldots  \hfill \mbox{\emph{classical variables}}\\ 
  r & ::= & r_0,r_1,\ldots   \hfill\mbox{\emph{quantum variables}}\\ 
  \pi & ::= & x\ \mid\ <x_1,\ldots,x_n>  \hfill \mbox{\emph{linear patterns}}\\ 
  \psi & ::= & \pi\ \mid\ !x  \hfill \mbox{\emph{patterns}}\\ 
  B & ::= & 0\ \mid\ 1\   \hfill \mbox{\emph{boolean constants}}\\ 
  U & ::= & U_0, U_1,\ldots   \hfill \mbox{\emph{unitary operators}}\\ 
  C & ::= & B\ \mid\ U  \hfill\mbox{\emph{constants}}\\ 
  M & ::= & x\ \mid\ r\ \mid\; !M\ \mid C \mid \new{M}\ \mid M_1M_2\ \mid \\ 
  & & \meas{M}\ \mid  \ifte{N}{M_1}{M_2}\ \mid \\
  & &  <M_1,\ldots,M_n>\ \mid \lambda \psi.M\  \mid
 \qquad \qquad \qquad\qquad\qquad \hfill \mbox{\emph{terms (where   $n\geq 2$)}}
\end{array}$$

We assume to work modulo variable renaming, i.e. \textit{terms are
equivalence classes modulo $\alpha$-conversion}. Substitution up to
$\alpha$-equivalence is defined in the usual way.
Since we are working modulo $\alpha$-conversion, we are authorized to use the so called Barendregt 
Convention on Variables (shortly, BCV)~\cite{Bar84}: \textit{in each mathematical context (a term, 
a definition, a proof...) the names chosen for bound variables will always differ from those of the free ones.}

Let us denote with $\Qvt{M_1,\ldots,M_k}$ the set of quantum variables
occurring in $M_1,\ldots,M_k$.
Notice that:
\begin{varitemize}
\item Variables are either \emph{classical} or \emph{quantum}: the
  first ones are the usual variables of lambda calculus (and can be bound by
  abstractions), while each
  quantum variable refers to a qubit in the underlying quantum
  register (to be defined shortly).
\item There are two sorts of constants as well, namely \emph{boolean
    constants} ($0$ and $1$) and \emph{unitary operators}: the first
  ones are useful for generating qubits and play no significant r\^ole in classical
  computation, while unitary operators are applied to (tuples of)
  quantum variables when performing quantum computation.
\item
  The term constructor $\new{\cdot}$ creates a new qubit when applied to a 
  boolean constant.
  \item The term constructor $\mathsf{meas}(\cdot)$ perform a single qubit measurement when applied to a quantum variable. 
\item The syntax allows the so
  called pattern abstraction.  A pattern is either a classical variable,  a tuple of
  classical variables, or a ``banged'' variable (namely an expression $!x$, where $x$ is a classical variable). In order
  to allow an abstraction of the kind $\lambda!x.M$, environments (see below)
  can contain $!$--patterns, denoting duplicable or erasable variables.
\end{varitemize}

For each qvs $\QV$ and for each quantum variable $r\in \QV$, we assume to have two, measurement based, linear 
transformation  of quantum registers:
$\mis{r}{0}, \mis{r}{1} :
\Hs{\QV} \to \Hs{\QV-\{r\}}$ (see Section~\ref{sect:qrm} for more details).

Judgements are defined from various notions of environments,  that
take into account the way the variables are used. Following
common notations in type theory and proof theory, a set of
variables $\{x_1,\ldots,x_n\}$ is often written simply as
$x_1,\ldots,x_n$. Analogously, the union of two sets of variables
$X$ and $Y$ is denoted simply as $X,Y$.
\begin{varitemize}
\item 
  A \textit{classical environment} is a (possibly empty) set of classical
  variables. Classical environments are denoted by $\Delta$ (possibly with indexes). 
  Examples of classical environments are $x_1,x_2$ or $x,y,z$ or the empty
  set $\emptyset$. Given
  a classic environment $\Delta=x_1,\ldots,x_n$,
  $!\Delta$ denotes the set of patterns $!x_1,\ldots,!x_n$.
\item 
  A \textit{quantum environment} is a (possibly empty) set (denoted by
  $\Theta$, possibly indexed) of quantum variables.
  Examples of quantum environments are $r_1,r_2,r_3$ and the empty
  set $\emptyset$.
\item 
  A \textit{linear environment} is (possibly empty) set (denoted by
  $\Lambda$, possibly indexed) in the form $\Delta,\Theta$
  where $\Delta$ is a classical environment and $\Theta$ is a quantum
  environment. The set $x_1,x_2,r_1$ is an example of a linear
  environment.
\item 
  An \textit{environment} (denoted by $\Gamma$, possibly indexed) is
  a (possibly empty) set in the form $\Lambda,!\Delta$ where each classical
  variable $x$ occurs at most once (either as $!x$ or as $x$) in $\Gamma$.
  For example, $x_1,r_1,!x_2$ is an environment, while $x_1,!x_1$ is \emph{not}
  an environment.
\item 
  A \emph{judgement} is an expression $\Gamma\vdash M$, where $\Gamma$
  is an environment and $M$ is a term.
\end{varitemize}

\begin{figure*}[!htb]
\dlinea
{\footnotesize
$$
\begin{array}{r@{\qquad}c@{\qquad}c@{\qquad}l} 
  \urule{}{!\Delta\vdash C}{\mbox{\textsf{const}}} &
  \urule{}{!\Delta, r\vdash r}{\mbox{\textsf{qvar}}} & 
  \urule{}{!\Delta,  x \vdash x}{\mbox{\textsf{cvar}}} &
  \urule{}{!\Delta ,!x\vdash x}{\mbox{\textsf{der}}} 
\end{array}
$$
$$
\begin{array}{r@{\qquad}c@{\qquad}l}
  \urule{!\Delta\vdash M}{!\Delta\vdash !M}{\mbox{\textsf{prom}}} &
  \brule{\Lambda_1,!\Delta\vdash M}{\Lambda_2, !\Delta\vdash N}
    {\Lambda_1,\Lambda_2, !\Delta\vdash MN}{\mbox{\textsf{app}}} &
  \urule{\Lambda_1,!\Delta\vdash M_{1} \cdots \Lambda_k,!\Delta\vdash M_{k}}
  {\Lambda_1,\ldots,\Lambda_k, !\Delta \vdash <M_{1},\ldots, M_{k}>}
  {\mbox{\textsf{tens}}}
\end{array}
$$
$$
\begin{array}{r@{\qquad}r@{\qquad}r@{\qquad}l} 
  \urule{\Gamma\vdash M}{\Gamma\vdash \new{M}}{\mathtt{\mbox{\textsf{new}}}} &
  \urule{\Gamma ,x_1,\ldots,x_n\vdash M}{\Gamma \vdash \lambda <x_1,\ldots,x_n>.M}{\textsf{lam1}} &
  \urule{\Gamma ,x\vdash M}{\Gamma \vdash \lambda  x.M}{\textsf{lam2}} &
  \urule{\Gamma ,!x\vdash M}{\Gamma \vdash \lambda ! x.M}{\textsf{lam3}} 
\end{array}
$$
$$
\begin{array}{r@{\qquad}r@{\qquad}r@{\qquad}l} 
\urule{\Gamma \vdash M}{\Gamma \vdash \meas{M}}{\textsf{ meas}} &
  \trule{\Lambda\vdash N}{!\Delta\vdash M_1}{!\Delta\vdash M_2}{\Lambda,!\Delta \vdash \ifte{N}{M_1}{M_2}}{\textsf{ if}} 
\end{array}
$$}
\dlinea
\caption{Well--Forming Rules}\label{fig:wfr}
\end{figure*}
We say that a judgement $\Gamma\vdash
M$ is \textit{well--formed} (\textit{notation:} $\wfj{}\Gamma\vdash M$)
if it is derivable from the \textit{well--forming rules} in
Figure~\ref{fig:wfr}.
The rules $\mathsf{app}$ and $\mathsf{tens}$ are subject to the following
constraint: for each $i\neq j$ $\Lambda_i\cap\Lambda_j=\emptyset$
(notice that $\Lambda_i$ and $\Lambda_j$ are sets of linear and
quantum variables, being linear environments).
With $\wfj{d}\Gamma\vdash M$ we mean that $d$ is a
derivation of the well--formed judgement $\Gamma\vdash M$.
If $\Gamma\vdash M$ is well--formed, the term $M$ is said to
be \emph{is well--formed with respect to the environment $\Gamma$}.
A term $M$ is \emph{well--formed} if the judgement
$\Qvt{M}\vdash M$ is well--formed.

\begin{remark}
\qstar\ comes equipped with two constants $0$ and $1$, and an $\ifte{(\cdot)}{(\cdot)}{(\cdot)}$ constructor.
However, these constructors can be thought of as syntactic sugar. Indeed,
$0$ and $1$ can be encoded as pure terms:
$0=\lambda !x.\lambda !y.y$ and 
$1=\lambda !x.\lambda !y.x $.
In doing so, $\ifte{M}{N}{L}$ becomes $M!N!L$.
The well--forming rule \textsf{if} (see Figure~\ref{fig:wfr}) of \qstar\ fully agrees with the above encodings.
\end{remark}

\condinc
{
\section{Quantum Registers and Measurements}\label{sect:qrm}
Before giving the definition of destructive measurement used in this paper we must 
clarify something about quantum spaces.

The smallest quantum space is $\Hs{\emptyset}$, which is (isomorphic to) 
the field $\CC$. The so called \textit{empty quantum register} is nothing 
more than a unitary element of $\CC$ (i.e., a complex number $c$ such that
$|c|=1$). We have chosen the scalar number $1$ as the canonical 
empty quantum register. In particular the number $1$ represents also the 
computational basis of $\Hs{\emptyset}$.

It is easy to show that if $\QV\cap \RV=\emptyset$ then there is a
standard \textit{isomorphism} 
$$
\Hs{\QV}\otimes\Hs{\RV} \stackrel{{i}_s}{\simeq} \Hs{\QV\cup\RV}.
$$
In the rest of this paper we will assume
to work up-to such an isomorphism\footnote{
in particular, if 
$\Qr\in \Hs{\QV}$, $r\not\in\QV$ and $|r\mapsto c>\in \Hs{\{r\}}$ 
then $\Qr\otimes|r\mapsto c>$ will denote the element 
$i_s(\Qr\otimes|r\mapsto c>)\in\Hs{\QV\cup\{r\}}$
}.
Note that the previous isomorphism holds even if either $\QV$ or $\RV$ is empty.

Since a quantum space $\Hs{\QV}$ is an Hilbert space, $\Hs{\QV}$ has a zero element $0_{\QV}$ 
(we will omit the subscript, when this does not cause ambiguity).
In particular, if $\QV\cap \RV=\emptyset$, 
$\Qr\in \Hs{\QV}$ and $\Rr \in \Hs{\RV}$, then 
$\Qr\otimes 0_{\RV} = 0_{\QV}\otimes\Rr = 0_{\QV\cup\RV}\in\Hs{\QV\cup\RV}$.

\begin{definition}[Quantum registers]
  Given a quantum space $\Hs{\QV}$, a \emph{quantum register} is any $\Qr\in\Hs{\QV}$ 
  such that either $\Qr=0_{\QV}$ or $\Qr$ is a normalised vector. 
\end{definition}

Let  $\QV$ be a qvs with cardinality $n\geq 1$. Moreover, let $\Qr\in\Hs{\QV}$ and let $r\in\QV$.+
Each state $\Qr$ may be represented as follows:
$$
\Qr=\sum_{i=1}^{2^{n-1}}\alpha_{i}|r\mapsto 0>\otimes b_i + \sum_{i=1}^{2^{n-1}}\beta_{i}|r\mapsto 1>\otimes b_i 
$$
where $\{b_i\}_{i\in[1,2^{n-1}]}$ is the computational basis\footnote{in the mathematical 
literature, computational basis are usually called standard basis; see~\cite{DLMZmscs08}, 
for the definition of computational/standard basis of $\Hs{\QV}$.}
of $\Hs{\QV-\{r\}}$. Please note that if $\QV=\{r\}$, then $\Qr=\alpha |r\mapsto 0>\otimes 1+\beta |r\mapsto 1>\otimes 1$, 
that is, via the previously stated isomorphism,   $\alpha |r\mapsto 0>+\beta |r\mapsto 1>$.

\begin{definition}[Destructive measurements]
Let  $\QV$ be a  qvs with cardinality $n=|\QV|\geq 1$,  
$r\in\QV$, $\{b_i\}_{i\in[1,2^{n-1}]}$ be the computational basis of $\Hs{\QV-\{r\}}$ 
and  $\Qr$ be $\sum_{i=1}^{2^{n-1}}\alpha_{i}|r\mapsto 0>\otimes b_i + \sum_{i=1}^{2^{n-1}}\beta_{i}|r\mapsto 1>\otimes b_i\in\Hs{\QV}$.
The two linear functions
$$
\mismin{r}{0},  
\mismin{r}{1} :
  \Hs{\QV} \to \Hs{\QV-\{r\}}
$$ 
such that
$$
\mismin{r}{0}(\Qr)=\sum_{i=1}^{2^{n-1}}\alpha_{i} b_i
\qquad 
\mismin{r}{1}(\Qr)=\sum_{i=1}^{2^{n-1}}\beta_{i} b_i
$$
are called \emph{destructive measurements}.
If $\Qr$ is a quantum register, the \emph{probability $p_c$ of observing $c\in\{0,1\}$ when observing $r$ in $\Qr$} is defined as
$\langle\Qr|{\mismin{r}{c}}^\dag\mismin{r}{c}|\Qr\rangle$.
\end{definition}
The just defined measurement operators are
\textit{general measurements}~\cite{KLM07,NieCh00}:
\begin{proposition}[Completeness Condition]
  Let $r\in\QV$ and $\Qr\in\Hs{\QV}$. Then
$\mismin{r}{0}^{\dag}\mismin{r}{0}+\mismin{r}{1}^{\dag}\mismin{r}{1}=Id_{\Hs{\QV}}$.
\end{proposition}
\begin{proof}
In order to prove the proposition we will use the following general property of inner product spaces: 
let $\HS$ be an inner product space and let $A:\HS\to\HS$ be a linear map.  If for each $x,y\in\HS$, $<Ax,y>= <x,y>$ 
then $A$ is the identity map\footnote{such a property is an immediate consequence of the \textit{Riesz representation theorem}, see e.g.~\cite{RomanBook}}.
Let $\Qr,\Rr\in\Hs{\QV}$. If $\{b_i\}_{i\in[1,2^{n}]}$ is the computational basis of $\Hs{\QV-\{r\}}$,  then:
\begin{align*}
\Qr&= \sum_{i=1}^{2^{n}}\alpha_{i}|r\mapsto 0>\otimes b_i+\sum_{i=1}^{2^{n}}\beta_{i}|r\mapsto 1>\otimes b_i \\
\Rr&= \sum_{i=1}^{2^{n}}\gamma_{i}|r\mapsto 0>\otimes b_i + \sum_{i=1}^{2^{n}}\delta_{i}|r\mapsto 1>\otimes b_i.
\end{align*}
We have:
\begin{align*}
\langle(\mismin{r}{0}^{\dag}\mismin{r}{0}+\mismin{r}{1}^{\dag}\mismin{r}{1})(\Qr), \Rr\rangle &=
  \langle\mismin{r}{0}^\dag\mismin{r}{0}(\Qr), \Rr\rangle+\langle\mismin{r}{1}^\dag\mismin{r}{1} (\Qr),\Rr\rangle\\
&=\langle\mismin{r}{0} (\Qr),\mismin{r}{0}(\Rr)\rangle + \langle\mismin{r}{1} (\Qr),\mismin{r}{1}(\Rr)\rangle\\
&=\langle\sum_{i=1}^{2^{n}}\alpha_{i}  b_i ,\sum_{i=1}^{2^{n}}\gamma_{i} b_i\rangle+\langle\sum_{i=1}^{2^{n}}\beta_{i}  b_i ,\sum_{i=1}^{2^{n}}\delta_{i} b_i\rangle\\
&=\sum_{i=0}^{2^{n}}\alpha_i\gamma_i +\sum_{i=0}^{2^{n}}\beta_i\delta_i \\
&=\langle \Qr,\Rr\rangle.
\end{align*}
This concludes the proof.
\end{proof}
For $c\in\{0,1\}$, the measurement operators $\mismin{r}{c}$ enjoys the following properties: 
\begin{proposition}\label{proposition:measCompl}
Let $\Qr\in\Hs{\QV}$. Then:
\begin{varenumerate}
\item\label{case:secondo} 
  $\mismin{r}{c}(\Qr\otimes|q\mapsto d>)
  = (\mismin{r}{c}(\Qr))\otimes|q\mapsto d>$ if $r\in\QV$ and $q\notin\QV$;
\item\label{case:terzo} 
  $\langle \Qr\otimes|s\mapsto d> |\mismin{r}{c}^\dag\mismin{r}{c}| \Qr\otimes|s\mapsto d>\rangle = 
  \langle \Qr,\mismin{r}{c}^\dag\mismin{r}{c}|\Qr\rangle$; if $r\in\QV$ and $r\neq s$;
\item 
  $\mismin{q}{e}(\mismin{r}{d}(\Qr)) =\mismin{r}{d}(\mismin{q}{e}(\Qr))$; if $r,q\in\QV$.
\end{varenumerate}
\end{proposition}
\begin{proof}
\begin{varenumerate}
\item 
Given the computational basis $\{b_i\}_{i\in[1,2^{n}]}$ of $\Hs{\QV-\{r\}}$, we have that:
$$
\Qr\otimes|q\mapsto d>=\sum_{i=1}^{2^{n}}\alpha_{i}|r\mapsto 0>\otimes b_i\otimes|q\mapsto d> + \sum_{i=1}^{2^{n}}\beta_{i}|r\mapsto 1>\otimes b_i|r\mapsto d>
$$
and therefore 
\begin{align*}
\mismin{r}{0}(\Qr\otimes|q\mapsto d>)&=\sum_{i=1}^{2^{n}}\alpha_{i}(b_i\otimes|r\mapsto d>)\\
 &=\left(\sum_{i=1}^{2^{n}}\alpha_{i} b_i\right) \otimes|q\mapsto d>\\
 &=(\mismin{r}{c}(\Qr))\otimes|q\mapsto d>.
\end{align*}
In the same way we prove the equality for $\mismin{r}{1}$.
\item  
Just observe that:
\begin{align*}
  \langle \Qr\otimes|s\mapsto d> |\mismin{r}{c}^\dag\mismin{r}{c}| \Qr\otimes|s\mapsto d>\rangle&=
  \langle \Qr\otimes |s\mapsto d>, \mismin{r}{c}^\dag(\mismin{r}{c}(\Qr\otimes|s\mapsto d>))\rangle\\ 
  &=\langle \mismin{r}{c}(\Qr\otimes|s\mapsto d>) , \mismin{r}{c}(\Qr\otimes|s\mapsto d>)\rangle\\
  &=\langle\mismin{r}{c}(\Qr) , \mismin{r}{c}(\Qr)\rangle\\
  &=\langle \Qr ,  \mismin{r}{c} ^\dag\mismin{r}{c} \Qr\rangle
   =\langle \Qr | \mismin{r}{c} ^\dag\mismin{r}{c} | \Qr\rangle .
\end{align*}
\item 
Given the computational basis $\{b_i\}_{i\in[1,2^{n}]}$ of $\Hs{\QV-\{r,q\}}$, we have that:
\begin{align*}
\Qr=&\sum_{i=1}^{2^{n}}\alpha_{i}|r\mapsto 0>\otimes|q\mapsto 0>\otimes b_i + \sum_{i=1}^{2^{n}}\beta_{i}|r\mapsto 0>\otimes|q\mapsto 1>\otimes b_i+\\
  &\sum_{i=1}^{2^{n}}\gamma_{i}|r\mapsto 1>\otimes|q\mapsto 0>\otimes b_i + \sum_{i=1}^{2^{n}}\delta_{i}|r\mapsto 1>\otimes|q\mapsto 1>\otimes b_i .
\end{align*}
Let us show that $\mismin{q}{0}(\mismin{r}{0}(\Qr)) = \mismin{r}{0}(\mismin{q}{0}(\Qr))$, the proof of  other cases follow the same pattern.
\begin{align*}
\mismin{r}{0}(\mismin{q}{0}(\Qr))&= 
\mismin{r}{0}\left(\sum_{i=1}^{2^{n}}\alpha_{i}|r\mapsto 0>\otimes b_i + \sum_{i=1}^{2^{n}}\gamma_{i}|r\mapsto 1>\otimes b_i\right)\\
 &=\sum_{i=1}^{2^{n}}\alpha_{i} b_i=\mismin{q}{0}\left(\sum_{i=1}^{2^{n}}\alpha_{i}|q\mapsto 0>\otimes b_i + \sum_{i=1}^{2^{n}}\beta_{i}|q\mapsto 1>\otimes b_i\right)\\
 &=\mismin{q}{0}(\mismin{r}{0}(\Qr)) .
\end{align*}
\end{varenumerate}
\noindent This concludes the proof.
\end{proof}
Given a qvs $\QV$ and a variable $r\in\QV$, we can define two linear maps:
$$
\mis{r}{0}, \mis{r}{1}:\Hs{\QV} \to \Hs{\QV-\{r\}}
$$
which are ``normalized'' versions of $\mismin{r}{0}$ and $\mismin{r}{1}$
as follows:
\begin{varenumerate}
\item if $\langle \Qr  | \mismin{r}{c}^\dag\mismin{r}{c}|\Qr\rangle = 0$ then $\mis{r}{c}(\Qr)=\mismin{r}{c}(\Qr)$;
\item if $\langle \Qr | \mismin{r}{c}^\dag\mismin{r}{c}| \Qr\rangle \neq 0$ then
$\mis{r}{c}(\Qr)=\frac{\mismin{r}{c}(\Qr)}{\sqrt{\langle \Qr | \mismin{r}{c}^\dag\mismin{r}{c}| \Qr\rangle }}$.
\end{varenumerate}
\begin{proposition}\label{Proposition:MisProp}
Let $\Qr\in\Hs{\QV}$ be a quantum register. Then:
\begin{varenumerate}
\item\label{qr-clos} 
  $\mis{r}{c}(\Qr)$ is a quantum register;
\item\label{primo} 
  $\mis{q}{e}(\Qr\otimes|r\mapsto d>)= (\mis{q}{e}(\Qr))\otimes|r\mapsto d>$, with $q\in\QV$ and $q\neq r$;
\item\label{secondo} 
  $\mis{q}{e}(\mis{r}{d}(\Qr)) =\mis{r}{d}(\mis{q}{e}(\Qr))$, with $q,r\in\QV$;
\item\label{secondobis} 
  if $q,r\in\QV$, $p_{r,c}=\langle \Qr | \mismin{r}{c}^{\dag}\mismin{r}{c}|\Qr\rangle$, 
  $p_{q,d}=\langle \Qr | \mismin{q}{d}^{\dag}\mismin{q}{d}|\Qr\rangle$,  $\Qr_{r,c}=\mis{r}{c}(\Qr)$, 
  $\Qr_{q,d}=\mis{q}{d}(\Qr)$, $s_{r,c}=\langle \Qr_{q,d} | \mismin{r}{c}^{\dag}\mismin{r}{c}|\Qr_{q,d}\rangle$, 
  $s_{q,d}=\langle \Qr_{r,c} | \mismin{q}{d}^{\dag}\mismin{q}{d}|\Qr_{r,c}\rangle$ 
  then $p_{r,c}\cdot s_{q,d}=p_{q,d}\cdot s_{r,c}$;
\item\label{terzo} 
  $(\mathbf{U}_{\langle q_{1},\ldots, q_{k}\rangle}\otimes \mathbf{I}_{\QV-\{q_{1},\ldots, q_{k}\}})(\mis{r}{c}(\Qr))=
  \mis{r}{c}((\mathbf{U}_{\langle q_{1},\ldots, q_{k}\rangle}\otimes \mathbf{I}_{\QV-\{q_{1},\ldots, q_{k}\}})(\Qr))$ 
  with  $\{q_{1}, \ldots, q_{k}\}\subseteq\QV$ and $r\neq q_{j}$ for all $j=1,\ldots, k$.
\end{varenumerate}
\end{proposition}
\begin{proof}
The proofs of \ref{qr-clos}, \ref{primo} and \ref{terzo} are immediate consequences of 
Proposition~\ref{proposition:measCompl} and of general basic properties of Hilbert spaces.
About \ref{secondo} and \ref{secondobis}:
if $\Qr=0_{\QV}$ then the proof is trivial;
if either $p_{r,c}=0$ or $p_{q,d}=0$ (possibly both), observe that 
$s_{r,c}=s_{q,d}=0$ and $\mis{q}{e}(\mis{r}{d}(\Qr)) =\mis{r}{d}(\mis{q}{e}(\Qr))=0_{\QV-\{q,r\}}$ and conclude.
Suppose now that $\Qr\neq 0_{\QV}$,  $p_{r,c}\neq 0$ and $p_{q,d} \neq 0$.
Given the computational basis $\{b_i\}_{i\in[1,2^{n}]}$ of $\Hs{\QV-\{r,q\}}$, we have that:
\begin{align*}
\Qr=&\sum_{i=1}^{2^{n}}\alpha_{i}|r\mapsto 0>\otimes|q\mapsto 0>\otimes b_i + \sum_{i=1}^{2^{n}}\beta_{i}|r\mapsto 0>\otimes|q\mapsto 1>\otimes b_i+\\
&\sum_{i=1}^{2^{n}}\gamma_{i}|r\mapsto 1>\otimes|q\mapsto 0>\otimes b_i + \sum_{i=1}^{2^{n}}\delta_{i}|r\mapsto 1>\otimes|q\mapsto 1>\otimes b_i
\end{align*}
\noindent Let us examine the case $c=0$ and $d=0$ (the other cases can be handled in the same way).
$$
p_{r,0} = \sum_{i=1}^{2^{n}}|\alpha_i|^2 +  \sum_{i=1}^{2^{n}}|\beta_i|^2 ;
\qquad \qquad
p_{q,0} = \sum_{i=1}^{2^{n}}|\alpha_i|^2 +  \sum_{i=1}^{2^{n}}|\gamma_i|^2;
$$
$$
\Qr_{r,0}=\mis{r}{0}(\Qr)= \frac{
\sum_{i=1}^{2^{n}}\alpha_{i}|q\mapsto 0>\otimes b_i + \sum_{i=1}^{2^{n}}\beta_{i}|q\mapsto 1>\otimes b_i
}
{
\sqrt{p_{r,0}}
}
$$
$$
\Qr_{q,0}= \mis{q}{0}(\Qr)=\frac{
\sum_{i=1}^{2^{n}}\alpha_{i}|r\mapsto 0>\otimes b_i + \sum_{i=1}^{2^{n}}\gamma_{i}|r\mapsto 1>\otimes b_i
}
{
\sqrt{p_{q,0}}
}
$$
Now let us consider the two states:
$$
\Qr_{r,0}^{q,0}=\mismin{q}{0}(\Qr_{r,0})= \frac{
\sum_{i=1}^{2^{n}}\alpha_{i} b_i
}
{
\sqrt{p_{r,0}} 
}
\qquad \qquad
\Qr_{q,0}^{r,0}=\mismin{r}{0}(\Qr_{q,0})= \frac{
\sum_{i=1}^{2^{n}}\alpha_{i} b_i
}
{
\sqrt{p_{q,0}} %
}
$$
By definition:
$$
s_{q,0} = \frac{\sum_{i=1}^{2^{n}}|\alpha_{i}|^2}{p_{r,0}}
\qquad 
s_{r,0} = \frac{\sum_{i=1}^{2^{n}}|\alpha_{i}|^2}{p_{q,0}}
$$
and therefore 
$p_{r,0}\cdot s_{q,d}=p_{q,0}\cdot s_{r,0}$.
\noindent Moreover, if $\QV=\emptyset$ then $\mis{q}{0}(\Qr_{r,0})=\mis{r}{0}(\Qr_{q,0})=1$, otherwise:
\begin{align*}
\mis{q}{0}(\Qr_{r,0})&=\frac{\Qr_{r,0}^{q,0}}{\sqrt{\bar{p}_{q,0}}}
  =\frac{\sum_{i=1}^{2^{n}}\alpha_{i} b_i}{\sqrt{p_{r,0}}\cdot \sqrt{\bar{p}_{q,0}}}
  =\frac{\sum_{i=1}^{2^{n}}\alpha_{i} b_i}{\sqrt{p_{r,0}}\cdot \sqrt{\frac{\sum_{i=1}^{2^{n}}|\alpha_{i}|^2}{p_{r,0}}}}
  =\frac{\sum_{i=1}^{2^{n}}\alpha_{i} b_i}{\sqrt{\sum_{i=1}^{2^{n}}|\alpha_{i}|^2}}\\
\mis{r}{0}(\Qr_{s,0})&=\frac{\Qr_{q,0}^{r,0}}{\sqrt{\bar{p}_{r,0}}}
  =\frac{\sum_{i=1}^{2^{n}}\alpha_{i} b_i}{\sqrt{p_{q,0}}\cdot \sqrt{\bar{p}_{r,0}}}
  =\frac{\sum_{i=1}^{2^{n}}\alpha_{i} b_i}{\sqrt{p_{q,0}}\cdot \sqrt{\frac{\sum_{i=1}^{2^{n}}|\alpha_{i}|^2}{p_{q,0}}}}
  =\frac{\sum_{i=1}^{2^{n}}\alpha_{i} b_i}{\sqrt{\sum_{i=1}^{2^{n}}|\alpha_{i}|^2}}
\end{align*}
and therefore $\mis{q}{0}(\Qr_{r,0})= \mis{r}{0}(\Qr_{s,0})$.
\end{proof}}{}
\subsection{Computations}\label{sect:comp}
In \qstar\ a computation is performed by reducing configurations.
A \emph{preconfiguration} is a triple $[\Qr,\QV,M]$ where:
\begin{varitemize}
\item $M$ is a term;
\item $\QV$ is a finite quantum variable set such that 
$\Qvt{M}\subseteq \QV$;
\item $\Qr\in\Hs{\QV}$.
\end{varitemize}
Let $\theta:\QV\rightarrow\RV$ be a bijective function from a (nonempty) 
finite set of quantum variables $\QV$ to another set of quantum
variables $\RV$.  Then we can extend $\theta$ to any term whose
quantum variables are included in $\QV$: $\theta(M)$ will be identical
to $M$, except on quantum variables, which are changed according to
$\theta$ itself. Observe that
$\Qvt{\theta(M)}\subseteq\RV$. Similarly, $\theta$ can be extended to
a function from $\Hs{\QV}$ to $\Hs{\RV}$ in the obvious way.

\begin{definition}[Configurations]
  Two preconfigurations $[\Qr,\QV,M]$ and $[\Rr,\RV,N]$ are
  equivalent iff there is a bijection $\theta:\QV\rightarrow\RV$ such
  that $\Rr=\theta(\Qr)$ and $N=\theta(M)$.  If a preconfiguration
  $C$ is equivalent to $D$, then we will write $C\equiv D$. The
  relation $\equiv$ is an equivalence relation.
  A \emph{configuration} is an equivalence class of preconfigurations
  modulo the relation $\equiv$. Let $\conf$ be the set of
  configurations.
\end{definition}

\begin{remark}
The way configurations have been defined, namely quotienting
preconfigurations over $\equiv$, is very reminiscent of usual
$\alpha$-conversion in lambda-terms.
\end{remark}

Let $\redrul=\{\Uq,\nw,\lbeta,\qbeta,\cbeta,\lcom,\rcom, \ifl, \ifr, \measr_r\}$.
For every $\alpha\in\redrul$ and for every $p\in\mathbb{R}_{[0,1]}$,
we define a relation $\ppredto{\alpha}{p}\subseteq\conf\times\conf$
by the set of \textit{contractions} in Figure ~\ref{fig:contractions}.   
The notation $\confone\predto{\alpha}\conftwo$ stands for
$\confone\ppredto{\alpha}{1}{\conftwo}$.

\begin{figure*}[!htb]
\dlinea
{\footnotesize 
$$
\begin{array}{c}%
\zrule
  { 
    [\Qr,\QV,(\lambda x.M)N]\ppredto{\lbeta}{1}
    [\Qr,\QV,M\{N/x\}]
  }
  {} 
\quad
\zrule
  { 
    [\Qr,\QV,(\lambda!
    x.M)!N]\predto{\cbeta}{1}[\Qr,\QV, M\{N/x\}] 
  } 
  {}
  \\[2ex]
  \zrule
  {
      [\Qr,\QV,(\lambda<x_1,\ldots,x_n>.M)<r_1,\ldots,r_n>]\ppredto{\qbeta}{1}
      [\Qr,\QV,M\{r_1/x_1,\ldots,r_n/x_n\}]
  } 
  {}
 \\[2ex]
\zrule
  {
      [\Qr,\QV,\ifte{1}{M}{N}]\ppredto{\ifl}{1}
      [\Qr,\QV,M]
  } 
  {}
  \\[2ex]
\zrule
  {
      [\Qr,\QV,\ifte{0}{M}{N}]\ppredto{\ifr}{1}
      [\Qr,\QV,N]
  } 
  {}

\\[2ex]
\zrule
  {
      [\Qr,\QV,U<r_{i_1},...,r_{i_n}>]\ppredto{\Uq}{1}
     {[\mathbf{U}_{<< r_{i_1},\ldots,r_{i_n}>>}\Qr,\QV,<r_{i_1},...,r_{i_n}>]}
  }
  {} 
  \\[2ex]
\zrule
  {
      [\Qr,\QV,\meas{r}]\ppredto{\measr_r}{p_c}
     {[\mis{r}{c}(\Qr),\QV-\{r\},!c]}
  }
  {} 
 \qquad(c\in\{0,1\} \mbox{ and } p_{c}=\langle\Qr|{\mismin{r}{c}}^\dag\mismin{r}{c}|\Qr\rangle\in\mathbb{R}_{[0,1]})\qquad
  \\[2ex]
 \zrule
  {
      [\Qr,\QV,\new{c}]\ppredto{\nw}{1}
      [\Qr\otimes\dr{r\mapsto c},\QV\cup\{r\}, r ]
  } 
  {} 
 \qquad(r \mbox{ is fresh})\qquad
\\[2ex]
 \zrule
  {
      [\Qr,\QV,L((\lambda \pi.M)N)]\ppredto {\lcom}{1}
      [\Qr,\QV,(\lambda\pi.LM)N]
  } 
  {}
  \\[2ex]
  \zrule
  {
      [\Qr,\QV,((\lambda \pi.M)N)L]\ppredto{\rcom}{1}
      [\Qr,\QV,(\lambda\pi.ML)N]
  } 
  {}
\\[2ex]
 \urule
 {
   {[\Qr,\QV,M]\ppredto{\alpha}{p}[\Rr,\RV,N]}
 }
 {
    \vseq{
   \qquad [\Qr,\QV,,<M_1,\ldots,M,\ldots,M_k>]\ppredto{\alpha}{p}
  {[\Rr,\RV,<M_1,\ldots,N,\ldots,M_k>]}
   \qquad
    } 
   }%
   {\tupla{i}}
   \\[4ex]
  \mbox{}\\

   \urule
  {
       [\Qr,\QV,N]\ppredto{\alpha}{p}[\Rr,\RV,P]
   } 
   {[\Qr,\QV,MN]\ppredto{\alpha}{p}[\Rr,\RV,MP]}
   {\rapp} \\
  \mbox{}\\
   \urule{
     [\Qr,\QV,M]\ppredto{\alpha}{p}[\Rr,\RV,P]
   } 
   {[\Qr,\QV,MN]\ppredto{\alpha}{p}[\Rr,\RV,PN]}
   {\lapp}
   \\[4ex]
     \mbox{}\\
  \urule{
     [\Qr,\QV,M]\ppredto{\alpha}{p}[\Rr,\RV,N]
   } 
   {[\Qr,\QV,\new{M}]\ppredto{\alpha}{p}[\Rr,\RV,\new{N}]}
   {\innew}
 \\[4ex]
  \urule{
     [\Qr,\QV,M]\ppredto{\alpha}{p}[\Rr,\RV,N]
   } 
   {[\Qr,\QV,\meas{M}]\ppredto{\alpha}{p}[\Rr,\RV,\meas{N}]}
   {\mathsf{in.meas}}
 \\[4ex]
  \urule{
     [\Qr,\QV,M]\ppredto{\alpha}{p}[\Rr,\RV,N]
   } 
   {[\Qr,\QV,\ifte{M}{L}{P}]\ppredto{\alpha}{p}[\Rr,\RV,\ifte{N}{L}{P}]}
   {\mathsf{in.if}}
 \\[4ex]
   \urule{[\Qr,\QV,M]\ppredto{\alpha}{p}[\Rr,\RV,N]}
   {[\Qr,\QV,(\lambda !x . M)]\ppredto{\alpha}{p}[\Rr,\RV,(\lambda !x.N)]}
   {\inlambda_1} 
   \quad
   \hfill
   \urule{[\Qr,\QV,M]\ppredto{\alpha}{p}[\Rr,\RV,N]}
   {[\Qr,\QV,(\lambda \pi . M)]\ppredto{\alpha}{p}[\Rr,\RV,(\lambda\pi.N)]}
   {\inlambda_2} 
\end{array}
$$}
\dlinea
\caption{Contractions.}
\label{fig:contractions}
\end{figure*}
In order to be consistent with the so-called non-cloning and non-erasing properties, 
we adopt surface reduction~\cite{Simpson05,DLMZmscs08}: reduction
is not allowed in the scope of any $!$ operator.
Furthermore, as usual, we also forbid reduction in $N$ and $P$
in the term $\ifte{M}{N}{P}$. Observe that contractions include
two commutative rules $\lcom$ and $\rcom$ (see Figure~\ref{fig:contractions}): they
come from \qcalc, where they were essential to get \emph{quantum standardization}~\cite{DLMZmscs08}.

We distinguish three particular subsets of $\redrul$, namely
$\commrul=\{\lcom,\rcom \}$,
$\noncommrul = \redrul-(\commrul\cup\{\measr_r\})$
and $\nonmeasrul=\redrul-\{\measr_r\}$.
In the following, we write $M\rightarrow_\alpha N$ meaning that
there are $\Qr$, $\QV$, $\Rr$ and $\RV$ such that
$[\Qr,\QV,M]\predto{\alpha}[\Rr,\RV,N]$.
Similarly for the notation $M\predto{\mathscr{S}}N$ where
$\mathscr{S}$ is a subset of $\redrul$.
} 
{ 
\section{A Brief Introduction to \qstar}\label{sec:QStarWfr}

In~\cite{DLMZmscs08} we have introduced a measurement--free, untyped quantum
$\lambda$--calculus, called \qcalc, based on the 
\textit{quantum data and classical control} paradigm (see
e.g.~\cite{Sel04c,SelVal06}). In this paper we generalize \qcalc\ by extending the class of terms
with a measurement operator, obtaining \qstar. Space limitations prevent
us from being exhaustive, and when needed, we will make reference to our paper~\cite{DLMZmscs08}
and to the literature.

\qstar\ is based on the
notion of a \textit{configuration}, namely a triple $[\Qr, \QV, M]$ where $\Qr$
is a \textit{quantum register}\footnote{the ``empty'' quantum register will be denoted with the scalar number $1$.}, 
$\QV$ is a finite set of names, called
\textit{quantum variables}, and $M$ is an untyped \textit{term} based on
the linear lambda-calculus defined by Wadler~\cite{Wad94} and Simpson~\cite{Simpson05}.
$\conf$ denotes the set of all such configurations.

Quantum registers are systems of $n$ qubits, that, mathematically
speaking, are normalized vectors of finite dimensional Hilbert
spaces. In particular, a quantum register $\Qr$ of a configuration
$[\Qr, \QV, M]$, is a normalized vector of the Hilbert space
$\ell^2(\{0,1\}^\QV)$, denoted here with 
$\Hs{\QV}.$\footnote{see ~\cite{DLMZmscs08} for a full discussion of $\Hs{\QV}$
  and ~\cite{RomanBook} for a general treatment of $\ell^2(S)$ spaces.}
Roughly speaking, the reader not familiar with Hilbert spaces could
think that  quantum variables are pointers to qubits in the
quantum register.

There are three kinds of operations on quantum registers: \textit{(i)}
the \textit{new} operation, responsible of the creation of qubits;
\textit{(ii)} \textit{unitary operators}: each unitary operator
$\mathbf{U}_{<< q_1,\ldots,q_n>>}$ corresponds to a pure quantum
operation acting on qubits with names $q_1,\ldots,q_n$
(mathematically, a unitary transform on the Hilbert space
$\Hs{\{q_1,\ldots,q_n\}}$, see~\cite{DLMZmscs08}); \textit{(iii)}
\textit{one qubit measurement} operations $\mis{r}{0}, \mis{r}{1}$
responsible of the probabilistic reduction of the quantum state plus
the destruction of the qubit referenced by $r$: given a
quantum register $\Qr\in\Hs{\QV}$, and a quantum variable name $r\in
\QV$, we allow the (destructive) measurement of the qubit with name
$r$.\condinc{}
 { \footnote{More precisely, for every quantum variable $r$ we assume
  the existence of two linear measurement operators, $\mis{r}{0}, \mis{r}{1} :
  \Hs{\QV} \to \Hs{\QV-\{r\}}$ enjoying the completeness condition
  ${\mis{r}{0}}^\dag\mis{r}{0} + {\mis{r}{1}}^\dag \mis{r}{1}=
  \mathit{id}_{\Hs{\QV}}$ and such that, given a quantum register
  $\Qr\in\Hs{\QV}$, the measurement of the qubit with name $r$ in
  $\Qr$ gives the outcome $c$ (with $c\in\{0,1\}$) with probability
  $p_c=\langle\Qr|{\mis{r}{c}}^\dag\mis{r}{c}|\Qr\rangle$ and produces
  the new quantum register $\frac{\mis{r}{c}\Qr}{\sqrt p_c}$;
  see~\cite{NieCh00,KLM07} for a detailed discussion of general pure
  measurements.}
}

The other main component of a configuration is a
\textit{term}. The set of terms is built from \textit{(i)} a denumerable set of
\textit{classical variables}, ranged over by $x, x_0, \ldots$; \textit{(ii)} a
denumerable set of \textit{quantum variables}, ranged over by $r,
r_0,\ldots$; \textit{(iii)} a finite or at most denumerable set of names corresponding to
\textit{unitary operators}; \textit{(iv)} the \textit{boolean constants}
$0, 1$ and \textit{(v)} the operators \texttt{new} and \texttt{meas}. An \textit{environment} 
$\Gamma$ is a (possibly
empty) finite set in the form $\Lambda,!\Delta$, where $\Lambda$ is a
(possibly empty) set of classical and quantum variables, and $!\Delta$
denote a (possibly empty) set of patterns $!x_1,\ldots,!x_n$. We
impose that in an environment, each classical variable $x$ occurs at
most once (either as $!x$ or as $x$).
A \emph{judgement} is an expression $\Gamma\vdash M$, where $\Gamma$
  is an environment and $M$ is a term. We say that a judgement is \textit{well-formed} if 
it is derivable by means of the \textit{well-forming rules} in Figure~\ref{fig:wfr}. 

\condinc{%
\subsection{Measurement}
For each quantum variable $r$ we assume
  to have two linear measurement operators, $\mis{r}{0}, \mis{r}{1} :
  \Hs{\QV} \to \Hs{\QV-\{r\}}$ enjoying the completeness condition
  ${\mis{r}{0}}^\dag\mis{r}{0} + {\mis{r}{1}}^\dag \mis{r}{1}=
  Id_{\Hs{\QV}}$ and such that, given a quantum register
  $\Qr\in\Hs{\QV}$, the measurement of the qubit with name $r$ in
  $\Qr$ gives the outcome $c$ (with $c\in\{0,1\})$ with probability
  $p_c=\langle\Qr|{\mis{r}{c}}^\dag\mis{r}{c}|\Qr\rangle$ and produces
  the new quantum register $\frac{\mis{r}{c}\Qr}{\sqrt p_c}$;
  see~\cite{NieCh00,KLM07} for a detailed discussion of general pure
  measurements
}%
{}

\begin{figure*}[!htb]
\dlinea
{\footnotesize
$$
\begin{array}{r@{\qquad}c@{\qquad}c@{\qquad}l} 
  \urule{}{!\Delta\vdash C}{\mbox{\textsf{const}}} &
  \urule{}{!\Delta, r\vdash r}{\mbox{\textsf{qvar}}} & 
  \urule{}{!\Delta,  x \vdash x}{\mbox{\textsf{cvar}}} &
  \urule{}{!\Delta ,!x\vdash x}{\mbox{\textsf{der}}} 
\end{array}
$$
$$
\begin{array}{r@{\qquad}c@{\qquad}l}
  \urule{!\Delta\vdash M}{!\Delta\vdash !M}{\mbox{\textsf{prom}}} &
  \brule{\Lambda_1,!\Delta\vdash M}{\Lambda_2, !\Delta\vdash N}
    {\Lambda_1,\Lambda_2, !\Delta\vdash MN}{\mbox{\textsf{app}}} &
  \urule{\Lambda_1,!\Delta\vdash M_{1} \cdots \Lambda_k,!\Delta\vdash M_{k}}
  {\Lambda_1,\ldots,\Lambda_k, !\Delta \vdash <M_{1},\ldots, M_{k}>}
  {\mbox{\textsf{tens}}}
\end{array}
$$
$$
\begin{array}{r@{\qquad}r@{\qquad}r@{\qquad}l} 
  \urule{\Gamma\vdash M}{\Gamma\vdash \new{M}}{\mathtt{\mbox{\textsf{new}}}} &
  \urule{\Gamma ,x_1,\ldots,x_n\vdash M}{\Gamma \vdash \lambda <x_1,\ldots,x_n>.M}{\textsf{lam}_1} &
  \urule{\Gamma ,x\vdash M}{\Gamma \vdash \lambda  x.M}{\textsf{lam}_2} &
  \urule{\Gamma ,!x\vdash M}{\Gamma \vdash \lambda ! x.M}{\textsf{lam}_3} 
\end{array}
$$
$$
\begin{array}{r@{\qquad}r@{\qquad}r@{\qquad}l} 
\urule{\Gamma \vdash M}{\Gamma \vdash \meas{M}}{\textsf{ meas}} &
  \trule{\Lambda\vdash N}{!\Delta\vdash M_1}{!\Delta\vdash M_2}{\Lambda,!\Delta \vdash \ifte{N}{M_1}{M_2}}{\textsf{ if}} 
\end{array}
$$}
\dlinea
\caption{Well--Forming Rules}\label{fig:wfr}
\end{figure*}

Let $\redrul=\{\Uq,\nw,\lbeta,\qbeta,\cbeta,\lcom,\rcom, \ifl, \ifr, \measr_r\}$.
For every $\alpha\in\redrul$ and for every $p\in\mathbb{R}_{[0,1]}$,
we define a relation $\ppredto{\alpha}{p}\subseteq\conf\times\conf$
by the set of rewriting rules \textit{contractions} in Figure ~\ref{fig:reduction}, plus
standard closure rules.   
The notation $\confone\predto{\alpha}\conftwo$ stands for
$\confone\ppredto{\alpha}{1}{\conftwo}$.
\begin{figure*}[!htb]
\dlinea
{\footnotesize
$$
\begin{array}{c}%
\zrule
  { 
    [\Qr,\QV,(\lambda x.M)N]\ppredto{\lbeta}{1}
    [\Qr,\QV,M\{N/x\}]
  }
  {} 
\quad
\zrule
  { 
    [\Qr,\QV,(\lambda!
    x.M)!N]\ppredto{\cbeta}{1}[\Qr,\QV, M\{N/x\}] 
  } 
  {}
  \\[2ex]
\zrule
  {
      [\Qr,\QV,(\lambda<x_1,\ldots,x_n>.M)<r_1,\ldots,r_n>]\ppredto{\qbeta}{1}
      [\Qr,\QV,M\{r_1/x_1,\ldots,r_n/x_n\}]
  } 
  {}
 \\[2ex]
\zrule
  {
      [\Qr,\QV,\ifte{1}{M_1}{M_2}]\ppredto{\ifl}{1}
      [\Qr,\QV,M_1]
  } 
  {}
  \\[2ex]
\zrule
  {
      [\Qr,\QV,\ifte{0}{M_1}{M_2}]\ppredto{\ifl}{1}
      [\Qr,\QV,M_2]
  } 
  {}

\\[2ex]
\zrule
  {
      [\Qr,\QV,U<r_{i_1},...,r_{i_n}>]\ppredto{\Uq}{1}
     {[\mathbf{U}_{<< r_{i_1},\ldots,r_{i_n}>>}Q,\QV,<r_{i_1},...,r_{i_n}>]}
  }
  {} 
  \\[2ex]
\zrule
  {
      [\Qr,\QV,\meas{r}]\ppredto{p_c}{\measr_r}
     {[\mis{r}{c}(\Qr),\QV-\{r\},!c]}
  }
  {} 
 \qquad(c\in\{0,1\} \mbox{ and } p_{c}\in\mathbb{R}_{[0,1]})\qquad
  \\[2ex]
 \zrule
  {
      [\Qr,\QV,\new{c}]\ppredto{\nw}{1}
      [\Qr\otimes\dr{r\mapsto c},\QV\cup\{r\}, r ]
  } 
  {} 
 \qquad(r \mbox{ is fresh})\qquad
\\[2ex]
 \zrule
  {
      [\Qr,\QV,L((\lambda \pi.M)N)]\ppredto{\lcom}{1}
      [\Qr,\QV,(\lambda\pi.LM)N]
  } 
  {}
  \\[2ex]
  \zrule
  {
      [\Qr,\QV,((\lambda \pi.M)N)L]\ppredto{\rcom}{1}
      [\Qr,\QV,(\lambda\pi.ML)N]
  } 
  {}
\end{array}
$$}
\dlinea
\caption{Contractions.}
\label{fig:reduction}
\end{figure*}
In order to be consistent with the so-called non-cloning and non-erasing properties, 
we adopt surface reduction~\cite{Simpson05,DLMZmscs08}: reduction
is not allowed in the scope of any $!$ operator.
Furthermore, as usual, we also forbid reduction in $N$ and $P$
in the term $\ifte{M}{N}{P}$. Observe that contractions include
two commutative rules $\lcom$ and $\rcom$ (see Figure~\ref{fig:reduction}): they
come from \qcalc, where they were essential to get \emph{quantum standardization}~\cite{DLMZmscs08}.
We distinguish three particular subsets of $\redrul$, namely
$\commrul=\{\lcom,\rcom \}$,
$\noncommrul = \redrul-(\commrul\cup\{\measr_r\})$
and $\nonmeasrul=\redrul-\{\measr_r\}$.
In the following, we write $M\rightarrow_\alpha N$ meaning that
there are $\Qr$, $\QV$, $\Rr$ and $\RV$ such that
$[\Qr,\QV,M]\predto{\alpha}[\Rr,\RV,N]$.
Similarly for the notation $M\predto{\mathscr{S}}N$ where
$\mathscr{S}$ is a subset of $\redrul$.
}
\section{The Confluence Problem: an Informal Introduction}\label{sec:ConfPb}
The confluence problem is central for any quantum $\lambda$-calculus 
with measurements, as stressed in the introduction.

Let us consider the following configuration: 
$$\confone= [\unoqr ,\emptyset,(\lambda !x. (\ifte{x}{0}{1}))(\meas{H(\nw(0))})].$$
If we focus on reduction sequences, it is easy to check that there are two different
reduction sequences starting with $\confone$, the first ending in the normal form
$[\unoqr ,\emptyset,0]$ (with probability $1/2$) and the second in the normal form
$[\unoqr ,\emptyset,1]$ (with probability $1/2$).
But if we reason with mixed states, the situation changes:
the mixed state $\{1:\confone\}$ (i.e., the mixed state assigning
probability $1$ to $\confone$ and $0$ to any other configuration)
rewrites \emph{deterministically} to
$\{1/2:[\unoqr,\emptyset,0], 1/2:[\unoqr ,\emptyset,1]\}$ (where both
$[\unoqr ,\emptyset,0]$ and 
$[\unoqr ,\emptyset,1]$ have probability $1/2$). So, confluence seems
to hold.

\paragraph{Confluence in Other Quantum Calculi.}
\noindent
Contrarily to the measurement-free case, the above notion of
confluence is \emph{not} an expected result for a quantum lambda calculus.
Indeed, it does not hold in the quantum lambda calculus $\lambda_{\mathit{sv}}$
 proposed by Selinger and Valiron~\cite{SelVal06}. In $\lambda_{\mathit{sv}}$,
it is possible to exhibit a configuration $\confone$ that
gives as outcome the distribution 
$\{1: [\unoqr ,\emptyset, 0]\}$ when reduced call-by-value and the
distribution $\{1/2:[\unoqr,\emptyset,0], 1/2:[\unoqr ,\emptyset,1]\}$
if reduced call-by-name. This is a \emph{real} failure of confluence, which
is there even if one uses probability distributions in place of configurations.
The same phenomenon cannot happen in \qstar\ (as we will show in Section~\ref{sec:StrongConfl}):
this fundamental difference can be traced back to another one: the linear
lambda calculus with surface reduction (on which \qstar\ is based)
enjoys (a slight variation on) the so-called diamond property~\cite{Simpson05}, while in usual, pure,
lambda calculus (on which $\lambda_{\mathit{sv}}$ is based) confluence
only holds in a weaker sense.

\paragraph{Finite or infinite rewriting?}
\noindent
In \qstar, an infinite computation can tend to a configuration which is essentially
different from the configurations in the computation itself.
For example, a configuration $\confone=[\unoqr ,\emptyset,M]$  can be built\footnote{%
$M\equiv (\mathsf{Y}!(\lambda!f.\lambda!x\ifte{x}{0}{f(\meas{H(\new{0})})}))(\meas{H(\new{0})})$,
where $\mathsf{Y}$ is a fix point operator.
}  
such that:
\begin{varitemize}
\item
  after a finite number of reduction steps $\confone$ rewrites to a distribution in the form
  $\{\sum_{1\lt i\leq n}\frac{1}{2^i}:[\unoqr ,\emptyset,0], 1-\sum_{1\lt i\leq n}\frac{1}{2^i}: \conftwo\}$
\item
  only after infinitely many reduction steps the distribution $\{1:[\unoqr ,\emptyset, 0]\}$
  is reached.
\end{varitemize}
Therefore finite probability distributions of finite configurations
could be obtained by means of infinite rewriting. 
We believe that the study of confluence for infinite computations is important.

\paragraph{Related Work.}
\noindent
In the literature, probabilistic rewriting systems have
been already analyzed.
For example, Bournez and Kirchner~\cite{BourKirchnerRTA02} have introduced the notion
of a probabilistic abstract rewriting system as a
structure $A=(|A|,[\cdot\leadsto\cdot])$ where $|A|$ is a set and
$[\cdot\leadsto\cdot]$ is a function from $|A|$ to $\mathbb{R}$ such that 
for every $a\in |A|$, $\sum_{b\in |A|}[a\leadsto b]$ is either $0$ or $1$. Then, they define
a notion of \textit{probabilistic confluence} for a PARS:
such a structure is probabilistically locally confluent iff the probability to be  
locally confluent, in a classical sense, is different from $0$.
Unfortunately, Bournez and Kirchner's analysis does not apply
to \qstar, since \qstar\ is \emph{not} a PARS. Indeed, the quantity 
$\sum_{b\in |A|}[a\leadsto b]$ can in general be any natural number. 
Similar considerations hold for the probabilistic lambda
calculus introduced by Di Pierro, Hankin and
Wiklicky in \cite{DiPHanWi05}.
\section{A Probabilistic Notion of Computation}\label{sec:Pcomp}
We represent computations as (possibly) infinite trees.
In the following, a (possibly) infinite tree $T$ will be an $(n+1)$-tuple $[R,T_1,\ldots, T_n]$, where $n\geq 0$, 
$R$ is the \emph{root} of $T$ and $T_1,\ldots,T_n$ are its \emph{immediate subtrees}.
\begin{definition}\label{def:probcomp}
  A set of (possibly) infinite trees $\mathscr{S}$ is said to be a \emph{set of probabilistic computations} if
  $\pcompone\in\mathscr{S}$ iff (exactly) one of the following three conditions holds:
  \begin{varenumerate}
  \item\label{clause:pcfc}
    $\pcompone=[\confone]$ and $\confone\in\conf$.
  \item
    $\pcompone=[\confone,\pcomptwo]$, where
    $\confone\in\conf$, $\pcomptwo\in\mathscr{S}$ has root $\conftwo$ and 
    $\confone\redto_{\nonmeasrul} \conftwo$
  \item
    $\pcompone=[(p,q,\confone),\pcomptwo,\pcompthree]$, where
    $\confone\in\conf$, $\pcomptwo,\pcompthree\in\mathscr{S}$ have
    roots $\conftwo$ and  $\confthree$, 
    $\confone\ppredto{meas_r}{p}\conftwo$,
    $\confone\ppredto{meas_r}{q}\confthree$ and $p,q\in\mathbb{R}_{[0,1]}$;
  \end{varenumerate}
  The set of all (respectively, the set of finite) probabilistic computations is the largest set $\mathscr{P}$
  (respectively, the smallest set $\mathscr{F}$) of probabilistic computations with respect to set inclusion.
  $\mathscr{P}$ and $\mathscr{F}$ exist because of the Knapster-Tarski Theorem.
\end{definition}
We will often say that the root of $\pcompone=[(p,q,\confone),\pcomptwo,\pcompthree]$
is simply $\confone$, slightly diverging from the above definition without any danger
of ambiguity.
\begin{definition}
  A probabilistic computation $\pcompone$ is \textit{maximal} if for every leaf $\confone$ in $\pcompone$, 
  $\confone\in\NF$. More formally, (sets of) maximal probabilistic computations can be defined as in
  Definition~\ref{def:probcomp}, where clause~\ref{clause:pcfc} must be restricted to
  $\confone\in\NF$.
\end{definition}

We can give definitions and proofs over \emph{finite} probabilistic computations
(i.e., over $\mathscr{F}$)
by ordinary induction. An example is the following
definition. Notice that the same is not true for arbitrary probabilistic definitions,
since $\mathscr{P}$ is not a well-founded set.

\begin{definition}
  Let $\pcompone\in\mathscr{P}$ be a probabilistic computation. 
  A finite probabilistic computation $\pcomptwo\in\mathscr{F}$ is a \emph{sub-computation}
  of $\pcompone$, written $\pcomptwo\sqsubseteq\pcompone$ iff
  one of the following conditions is satisfied:
  \begin{varitemize}
    \item
      $\pcomptwo=[\confone]$ and the root of $\pcompone$ 
      is $\confone$.
    \item
      $\pcomptwo=[\confone,\pcompthree]$,
      $\pcompone=[\confone,\pcompfour]$, and
      $\pcompthree\sqsubseteq\pcompfour$.
    \item
      $\pcomptwo=[(p,q,\confone),\pcompthree,\pcompfour]$,
      $\pcompone=[(p,q,\confone),\pcompsix,\pcompseven]$,
      $\pcompthree\sqsubseteq\pcompsix$ and
      $\pcompfour\sqsubseteq\pcompseven$.      
    \end{varitemize}  
\end{definition}

Let $\delta:\conf \to \{0,1\}$ be a function defined as follows: $\delta(\confone) = 0$ 
if the quantum register of $\confone$ is $0$, otherwise, $\delta(\confone) = 1$.

\paragraph{Quantitative Properties of Computations.}
The outcomes of a probabilistic computation $\pcompone$ are given 
by the configurations which appear as leaves of $\pcompone$. 
Starting from this observation, the following definitions formalize
some quantitative properties of probabilistic computations.
\condinc{%
For every \emph{finite} probabilistic computation $\pcompone$ and every $\confone\in\NF$
we define $\probmc{\pcompone}{\confone}\in\mathbb{R}_{[0,1]}$ 
by induction on the structure of $\pcompone$:
\begin{varitemize}
\item
  $\probmc{[\confone]}{\confone}= \delta(\confone)$;
\item 
  $\probmc{[\confone]}{\conftwo}=0$ whenever $\confone\neq\conftwo$;
\item
  $\probmc{[\confone,\pcompone]}{\conftwo}=\probmc{\pcompone}{\conftwo}$;
\item $\probmc{[(p,q,\confone),\pcompone,\pcomptwo]}{\conftwo}=p\probmc{\pcompone}{\conftwo}+q\probmc{\pcomptwo}{\conftwo}$;
\end{varitemize}
Similarly for $\numlv{\pcompone}{\confone}\leq\aleph_0$:
\begin{varitemize}
\item
  $\numlv{[\confone]}{\confone}=1$;
\item 
$\numlv{[\confone]}{\conftwo}=0$ whenever $\confone\neq\conftwo$;
\item
  $\numlv{[\confone,\pcompone]}{\conftwo}=\numlv{\pcompone}{\conftwo}$;. 
\item 
$\numlv{[(p,q,\confone),\pcompone,\pcomptwo]}{\conftwo}=\numlv{\pcompone}{\conftwo}+\numlv{\pcomptwo}{\conftwo}$.
\end{varitemize}
Informally, $\probmc{\pcompone}{\confone}$ is the probability of observing $\confone$ as a leaf
in $\pcompone$, and $\numlv{\pcompone}{\confone}$ is the number of times  $\confone$ appears
as a leaf in $\pcompone$.

The definitions above can be easily modified to get the probability of observing \emph{any} configuration
(in normal form) as a leaf in $\pcompone$, $\probmcany{\pcompone}$, or the number of times \emph{any} configuration appears
as a leaf in $\pcompone$, $\numlvany{\pcompone}$.
Since $\mathbb{R}_{[0,1]}$ and $\NN\cup\{\aleph_0\}$ are complete 
lattices (with respect to standard orderings), we extend the above notions to the case of \textit{arbitrary} 
probabilistic computations, by taking the least upper bound over all finite sub-computations.
If $\pcompone\in\mathscr{P}$ and $\confone\in\NF$, then:
\begin{varitemize}
\item $\probmc{\pcompone}{\confone}=\sup_{\pcomptwo\sqsubseteq\pcompone}\probmc{\pcomptwo}{\confone}$;
\item $\numlv{\pcompone}{\confone}=\sup_{\pcomptwo\sqsubseteq\pcompone}\numlv{\pcomptwo}{\confone}$;
\item $\probmcany{\pcompone}=\sup_{\pcomptwo\sqsubseteq\pcompone}\probmcany{\pcomptwo}$;
\item $\numlvany{\pcompone}=\sup_{\pcomptwo\sqsubseteq\pcompone}\numlvany{\pcomptwo}$.
\end{varitemize}
}%
{%
For every \emph{finite} probabilistic computation $\pcompone$ and every $\confone\in\NF$
we define $\probmc{\pcompone}{\confone}\in\mathbb{R}_{[0,1]}$ and 
$\numlv{\pcompone}{\confone}\leq \aleph_0$ by induction on the structure of $\pcompone$:
\begin{varitemize}
\item
  $\probmc{[\confone]}{\confone}=\numlv{[\confone]}{\confone}=1$ and
  $\probmc{[\confone]}{\conftwo}=\numlv{[\confone]}{\conftwo}=0$ whenever $\confone\neq\conftwo$. 
\item
  $\probmc{[\confone,\pcompone]}{\conftwo}=\probmc{\pcompone}{\conftwo}$ and
  $\numlv{[\confone,\pcompone]}{\conftwo}=\numlv{\pcompone}{\conftwo}$. 
\item
  $\probmc{[(p,\confone),\pcompone,\pcomptwo]}{\conftwo}=p\probmc{\pcompone}{\conftwo}+q\probmc{\pcomptwo}{\conftwo}$ and
  $\numlv{[(p,\confone),\pcompone,\pcomptwo]}{\conftwo}=\numlv{\pcompone}{\conftwo}+\numlv{\pcomptwo}{\conftwo}$.
\end{varitemize}
More informally, $\probmc{\pcompone}{\confone}$ is the probability of observing $\confone$ as a leaf
in $\pcompone$. On the other hand, $\numlv{\pcompone}{\confone}$ is the number of times $\confone$ appears
as a leaf in $\pcompone$.
The definitions above can be easily modified to get the probability of observing \emph{any} configuration
as a leaf in $\pcompone$, $\probmcany{\pcompone}$, or the number of times \emph{any} configuration appears
as a leaf in $\pcompone$, $\numlvany{\pcompone}$. 

In turn, the functions $\probmcf$ and $\numlvf$ on finite probabilistic computations above can be generalized to functions
on arbitrary probabilistic computations by taking the least upper bound over all finite sub-computations.
For example, if $\pcompone\in\mathscr{P}$ and $\confone\in\NF$, then
$$
\probmc{\pcompone}{\confone}=\sup_{\pcomptwo\sqsubseteq\pcompone}\probmc{\pcomptwo}{\confone}.
$$
Analogously,
$$
\numlvany{\pcompone}=\sup_{\pcomptwo\sqsubseteq\pcompone}\numlvany{\pcomptwo}.
$$
Both quantities above exists because $\mathbb{R}_{[0,1]}$ and $\NN\cup\{\aleph_0\}$ are complete 
lattices.
}
\condinc{
The following lemmas involve finite computations and can be prove by induction.
\begin{lemma}\label{lemma:monotone}
If $\pcompone\sqsubseteq\pcomptwo$, then
$\probmcany{\pcompone}\leq\probmcany{\pcomptwo}$
and $\numlvany{\pcompone}\leq\numlvany{\pcomptwo}$. Moreover,
$\probmc{\pcompone}{\confone}\leq\probmc{\pcomptwo}{\confone}$
and $\numlv{\pcompone}{\confone}\leq\numlv{\pcomptwo}{\confone}$
for every $\confone\in\NF$.
\end{lemma}
\begin{proof}
A trivial induction on $\pcompone$.
\end{proof}
\begin{lemma}\label{lemma:maximal}
If $\pcompone\sqsubseteq\pcomptwo$ and $\pcompone$ is
maximal, then $\pcomptwo$ is maximal.
\end{lemma}
\begin{proof}
A trivial induction on $\pcompone$.
\end{proof}
}{
}
\section{A Strong Confluence Result}\label{sec:StrongConfl} 
In this Section, we will \condinc{prove}{give} a strong confluence result in the
following form:\textit{ any two maximal probabilistic computations
  $\pcompone$ and $\pcomptwo$ with the same root have exactly the same
  quantitative and qualitative behaviour}, that is to say, the
following equations hold for every $\confone\in\NF$:
\condinc{\begin{eqnarray*}
\probmc{\pcompone}{\confone}&=&\probmc{\pcomptwo}{\confone};\\
\numlv{\pcompone}{\confone}&=&\numlv{\pcomptwo}{\confone};\\
\probmcany{\pcompone}&=&\probmcany{\pcomptwo};\\
\numlvany{\pcompone}&=&\numlvany{\pcomptwo}.
\end{eqnarray*}}
{$\probmc{\pcompone}{\confone}=\probmc{\pcomptwo}{\confone}$, 
$\numlv{\pcompone}{\confone}=\numlv{\pcomptwo}{\confone}$,
$\probmcany{\pcompone}=\probmcany{\pcomptwo}$, and
$\numlvany{\pcompone}=\numlvany{\pcomptwo}$.
}
\begin{remark}
Please notice that equalities like the ones above do \emph{not} even hold for the
ordinary lambda calculus. For example, the lambda term $(\lambda x.\lambda y.y)\Omega$ is the
root of two (linear) computations, the first having one leaf $\lambda y.y$ and the
second having no leaves. This is the reason why the confluence result
we prove here is dubbed as strong.
\end{remark}
Before embarking in the proof of the equalities above, let us spend a few words to explain
their consequences. The fact $\probmc{\pcompone}{\confone}=\probmc{\pcomptwo}{\confone}$
whenever $\pcompone$ and $\pcomptwo$ have the same root can be read as a confluence
result: the probability of observing $\confone$ is independent from the adopted strategy.
On the other hand, $\probmcany{\pcompone}=\probmcany{\pcomptwo}$ means that 
the probability of converging is not affected by the underlying strategy. 
The corresponding results on $\numlv{\cdot}{\cdot}$ and $\numlvany{\cdot}$
can be read as saying that the number of (not necessarily distinct) leaves in any probabilistic
computation with root $\confone$ does not depend on the strategy.

\condinc{%
\begin{lemma}[Uniformity]\label{lemma:uniformity}
  For every $M,N$ such that $M\predto{\alpha} N$, exactly one of the
  following conditions holds:
  \begin{varenumerate}
    \item\label{firstcase}
      $\alpha\neq\nw$ and $\alpha\neq\measr_r$ and there is a unitary transformation
      $U_{M,N}:\Hs{\Qvt{M}}\rightarrow\Hs{\Qvt{M}}$ 
      such that $[\Qr,\QV,M]\predto{\alpha}[\Rr,\RV,N]$
      iff $\Ok{[\Qr,\QV,M]}$, $\RV=\QV$ and
      $\Rr=(U_{M,N}\otimes I_{\QV-\Qvt{M}})\Qr$.
    \item\label{secondcase}
      $\alpha=\nw$ and there are a constant $c$ and a
      quantum variable $r$ such that 
      $[\Qr,\QV,M]\predto{\nw}[\Rr,\RV,N]$
      iff $\Ok{[\Qr,\QV,M]}$, $\RV=\QV\cup\{r\}$ and
      $\Rr=\Qr\otimes\dr{r\mapsto c}$. 
    \item\label{thirdcase}
      $\alpha=\measr_r$ and there are a constant
      $c$ and a probability $p_c\in\mathbb{R}_{[0,1]}$ such that
      $[\Qr,\QV,M]\ppredto{\measr_r}{p_c}[\Rr,\RV,N]$
      iff $\Ok{[\Qr,\QV,M]}$, $\Rr=\mis{r}{c}(\Qr)$ and
      $\RV=\QV-\{r\}$.
  \end{varenumerate}
\end{lemma} 
\begin{proof} 
  We go by induction on $M$. $M$ cannot be a variable nor a constant
  nor a unitary operator nor a term $!L$. If $M$ is an abstraction
  $\lambda\psi.L$, then $N\equiv \lambda\psi.P$, $L\predto{\alpha} P$
  and the thesis follows from the inductive hypothesis.
  If $M$ is $\meas{L}$ and $N$ is $\meas{P}$ then  
  $L\predto{\alpha} P$ and the thesis follows from the inductive hypothesis.
  Similarly if $M$ is $\new{L}$ and $N$ is $\new{P}$. And 
  again if $M$ is $\langle M_1,\ldots,L,\ldots,M_n\rangle$
  and $N$ is $\langle M_1,\ldots,P,\ldots,M_n\rangle$.
  If $M\equiv LQ$, then we distinguish a number of cases:
\begin{varitemize}
\item
  $N\equiv PQ$ and $L\predto{\alpha}P$. The thesis follows from the inductive hypothesis.
\item
  $N\equiv LS$ and $Q\predto{\alpha}S$. The thesis follows from the inductive hypothesis.
\item
  $L\equiv U$, $Q\equiv <r_{1},...,r_{n}>$ and $N\equiv <r_{1},...,r_{n}>$. 
  Then case~\ref{firstcase} holds. In particular,
  $\Qvt{M}=\{r_{1},...,r_{n}\} $ and $U_{M,N}=U_{<<r_{1},...,r_{n}>>}$.
\item
  $L\equiv\lambda x.R$ and $N=R\{Q/x\}$. Then case~\ref{firstcase} holds. In particular
  $U_{M,N}=I_{\Qvt{M}}$.
\item
  $L\equiv\lambda <x_1,\ldots,x_n>.R$, $Q=<r_1,\ldots,r_n>$ 
  and $N\equiv R\{r_1/x_1,\ldots,r_n/x_n\}$. Then case~\ref{firstcase} holds and
  $U_{M,N}=I_{\Qvt{M}}$.
\item
  $L\equiv\lambda !x.R$, $Q=!T$ 
  and $N\equiv R\{T/x\}$. Then case~\ref{firstcase} holds and $U_{M,N}=I_{\Qvt{M}}$.
\item
  $Q\equiv(\lambda \pi.R)T$ and $N\equiv(\lambda\pi.LR)T$. Then case~\ref{firstcase} holds 
  and $U_{M,N}=I_{\Qvt{M}}$.
\item
  $L\equiv(\lambda \pi.R)T$ and $N\equiv(\lambda\pi.RQ)T$. Then case~\ref{firstcase} holds 
  and $U_{M,N}=I_{\Qvt{M}}$.
\end{varitemize}
If $M\equiv \new{c}$ then $N$ is a quantum variable $r$ and case~\ref{secondcase} holds.
If $M\equiv \meas{r}$ then there are a constant $c$ and a probability $p_c$ 
such that $N$ is a term $!c$ and case~\ref{thirdcase} holds.
This concludes the proof.
\end{proof}
Notice that $U_{M,N}$ is always the identity function when performing 
classical reduction.
The following technical lemma will be useful when proving confluence: 
\begin{lemma}\label{lemma:paramuniformity}
Suppose $[\Qr,\QV,M]\predto{\alpha}[\Rr,\RV,N]$.
\begin{varenumerate}
\item\label{firstclaim}
  If $\Ok{[\Qr,\QV,M\{L/x\}]}$, then
  $$
  [\Qr,\QV,M\{L/x\}]\predto{\alpha}[\Rr,\RV,N\{L/x\}].
  $$
\item\label{secondclaim}
  If $\Ok{[\Qr,\QV,M\{r_1/x_1,\ldots,r_n/x_n\}]}$,
  then 
  $$
  [\Qr,\QV,M\{r_1/x_1,\ldots,r_n/x_n\}]\predto{\alpha}[\Rr,\RV,N\{r_1/x_1,\ldots,r_n/x_n\}].
  $$
\item\label{thirdclaim}
  If $x,\Gamma\vdash L$ and $\Ok{[\Qr,\QV,L\{M/x\}]}$,
  then 
  $$
  [\Qr,\QV,L\{M/x\}]\predto{\alpha}[\Rr,\RV,L\{N/x\}].
  $$
\end{varenumerate}
\end{lemma}
\begin{proof}
Claims~\ref{firstclaim} and~\ref{secondclaim} can be proved by induction on the
proof of  $[\Qr,\QV,M]\predto{\alpha}[\Rr,\RV,N]$. Claim~\ref{thirdclaim} can be
proved by induction on $N$.
\end{proof}
We prove now that \qstar\ enjoys a slight variation of the so-called diamond property, whose proof is fully standard 
(it is a slight extension of the analogous proof given in~\cite{DLMZmscs08} for \qcalc).
As for \qcalc, \qstar\ does not enjoy the diamond property in a strict sense, due to
the presence of commutative reduction rules (see, e.g., case 2 of the following Proposition).  
But thanks to Lemma~\ref{lemma:noinfcom} below, this does not have harmful consequences.

\begin{proposition}[Quasi-One-step Confluence]\label{prop:onestepconf}
  Let $\confone,\conftwo,\confthree$ be configurations with $\confone\ppredto{\alpha}{p}\conftwo$, $\confone\ppredto{\beta}{s}\confthree$. Then:
  \begin{varenumerate}
  \item
    If $\alpha\in\commrul$ and $\beta\in\commrul$, then either $\conftwo=\confthree$ or
    there is $F$ with $D\predto{\commrul}F$ and $E\predto{\commrul}F$. 
  \item 
    If $\alpha\in\commrul$ and $\beta\in\noncommrul$, then either 
    $D\predto{\noncommrul}E$ or there 
    is $F$ with $D\predto{\noncommrul} F$ and $E\predto{\commrul}F$.
  \item
    If $\alpha\in\commrul$ and $\beta=\measr_r$, then there is $\conffour$
    with $\conftwo\ppredto{\measr_r}{s}\conffour$ and
    $\confthree\predto{\commrul}\conffour$. 
 \item 
    If $\alpha\in\noncommrul$ and $\beta\in\noncommrul$, then either $\conftwo=\confthree$ or
    there is $F$ with $D\predto{\noncommrul}F$ and $E\predto{\noncommrul}F$.
  \item
    If $\alpha\in\noncommrul$ and $\beta=\measr_r$, then there is  $\conffour$
    with $\conftwo\ppredto{\measr_r}{s}\conffour$ and
    $\confthree\predto{\commrul}\conffour$.
 \item
    If $\alpha=\measr_r$ and $\beta=\measr_q$ ($r\neq q$), then there are $t,u\in\RR_{[0,1]} $ and a $\conffour$ such that 
     $pt=su$, $\conftwo\ppredto{\measr_q}{t}\conffour$ and
    $\confthree\ppredto{\measr_r}{u}\conffour$.
  \end{varenumerate}
\end{proposition}
\begin{proof}
Let $C\equiv [\Qr,QV,M]$. We go by induction on $M$. $M$ cannot be a variable 
nor a constant nor a unitary operator. If $M$ is an abstraction $\lambda\pi.N$, then
$D\equiv [\Rr,\RV,\lambda \pi.S]$, $E\equiv [\Sr,\SV,\lambda \pi.T]$ and
\begin{eqnarray*}
  [\Qr,\QV,N]&\predto{\alpha}&[\Rr,\RV,S]\\
  \ [\Qr,\QV,N]&\predto{\beta}&[\Sr,\SV,T]
\end{eqnarray*}
The IH easily leads to the thesis. Similarly when $M\equiv\lambda !x.N$, and  when $M\equiv\meas{N}$ or $M\equiv\ifte{N}{P}{Q}$ with $N\neq 0,1$.
If $M\equiv NL$, we can distinguish a number of cases depending on the 
last rule used to prove $C\ppredto{\alpha}{p}D$, $C\predto{\beta}{s}E$:
\begin{varitemize}
\item
  $D\equiv [\Rr,\RV,SL]$ and $E\equiv [\Sr,\SV,NR]$
  where 
  $[\Qr,\QV,N]\ppredto{\alpha}{p}[\Rr,\RV,S]$
  and $[\Qr,\QV,L]\ppredto{\beta}{s}[\Sr,\SV,R]$. 
  We need to distinguish several sub-cases:
  \begin{varitemize}
  \item
    If $\alpha,\beta=\nw$, then, by Lemma~\ref{lemma:uniformity}, there exist
    two quantum variables $s,t\notin\QV$ and two constants $d,e$ such that 
    $\RV=\QV\cup\{s\}$, $\SV=\QV\cup\{t\}$, 
    $\Rr=\Qr\otimes\dr{s\mapsto d}$
    and $\Sr=\Qr\otimes\dr{t\mapsto e}$.
    Applying~\ref{lemma:uniformity} again, we obtain
    \begin{eqnarray*}
      D&\predto{\nw}&[\Qr\otimes\dr{s\mapsto d}\otimes\dr{u\mapsto e},
      \QV\cup\{s,u\},SR\{u/t\}]\equiv F;\\
      E&\predto{\nw}&[\Qr\otimes\dr{t\mapsto e}\otimes\dr{v\mapsto d},
      \QV\cup\{t,v\},S\{u/s\}R]\equiv G.
    \end{eqnarray*}
    As can be easily checked, $F\equiv G$.
  \item
    If $\alpha=\nw$ and $\beta\neq\nw,\measr_r$, then, by Lemma~\ref{lemma:uniformity}
    there exist a quantum variable $r$ and a constant $c$ such that
    $\RV=\QV\cup\{r\}$, $\Rr=\Qr\otimes\dr{r\mapsto c}$,
    $\SV=\QV$ and $\Sr=(U_{L,R}\otimes I_{\QV-{\Qvt{L}}})\Qr$. As a consequence,
    applying Lemma~\ref{lemma:uniformity} again, we obtain
    \begin{eqnarray*}
      D&\predto{\beta}&[(U_{L,R}\otimes I_{\QV\cup\{r\}-{\Qvt{L}}})(\Qr
      \otimes\dr{r\mapsto c}),\QV\cup\{r\},SR]\equiv F;\\
      E&\predto{\nw}&[((U_{L,R}\otimes I_{\QV-{\Qvt{L}}})\Qr)
      \otimes\dr{r\mapsto c},\QV\cup\{r\},SR]\equiv G.
    \end{eqnarray*}
    As can be easily checked, $F\equiv G$.
  \item
    If $\alpha\neq\nw,\measr_r$ and $\beta=\nw$, then we can proceed as in the previous
    case.
  \item
    If $\alpha,\beta\neq\nw, \alpha\neq\measr_r, \beta\neq\measr_q$ ($r,q$ not necessarily distinct) , then by Lemma~\ref{lemma:uniformity}, there exist
    $\SV=\RV=\QV$, 
    $\Rr=(U_{N,S}\otimes I_{\QV-{\Qvt{N}}})\Qr$ and
    $\Sr=(U_{L,R}\otimes I_{\QV-{\Qvt{L}}})\Qr$.
    Applying~\ref{lemma:uniformity} again, we obtain
    \begin{eqnarray*}
      D&\predto{\beta}&[(U_{L,R}\otimes I_{\QV-{\Qvt{L}}})((U_{N,S}
      \otimes I_{\QV-{\Qvt{N}}})\Qr),\QV,SR]\equiv F;\\
      E&\predto{\alpha}&[(U_{N,S}\otimes I_{\QV-{\Qvt{L}}})((U_{L,R}
      \otimes I_{\QV-{\Qvt{L}}})\Qr),\QV,SR]\equiv G.
    \end{eqnarray*}
    As can be easily checked, $F\equiv G$.
   \item 
    If $\alpha =\measr_r,\beta=\measr_q$ ($r\neq q$) then, by Lemma~\ref{lemma:uniformity}, there exist
    two constants $d,e$ and two probabilities $t,u$ such that 
    $\RV=\QV-\{r\}$, $\SV=\QV-\{q\}$, 
    $\Rr=\mis{r}{d}(\Qr)$
    and $\Sr=\mis{q}{e}(\Qr)$.  Remember that the quantum variable $q$ occurs in the subterm $N$ and the quantum variable $r$ occurs in the subterm $L$.    
    Starting from  $D\equiv [\mis{r}{d}(\Qr), \QV-\{r\}, SL]$  and $E\equiv [\mis{q}{e}(\Qr), \QV-\{q\}, NR]$,    
    applying~\ref{lemma:uniformity} again, we obtain
    \begin{eqnarray*}
      D & \ppredto{\measr_q}{\bar{s}}& [\mis{q}{e}(\mis{r}{d}(\Qr)), \QV-\{r\}-\{q\},SR]\\
        &\equiv&[\mis{q}{e}(\Rr), \RV-\{q\}, SR]\equiv F;\\
      E  &\ppredto{\measr_r}{\bar{p}}& [\mis{r}{d}(\mis{q}{e}(\Qr)), \QV-\{q\}-\{r\},SR]\\
        &\equiv&[\mis{r}{d}(\Sr),\SV-\{r\},SR]\equiv G.
    \end{eqnarray*}    
    Clearly,   $\QV-\{r\}-\{q\}\equiv \QV-\{q\}-\{r\} $ and by 
    Proposition~\ref{Proposition:MisProp}, case~\ref{secondobis}, $\mis{q}{e}(\mis{r}{d}(\Qr))\equiv \mis{r}{d}(\mis{q}{e}(\Qr))$.
    Then  $F\equiv G$. Moreover by Proposition~\ref{Proposition:MisProp}, case~\ref{secondo},   $pt=su$.
  \item 
    If $\alpha =\nw, \beta=\measr_r$, then, by Lemma~\ref{lemma:uniformity}
    there exists a quantum variable $q$ ($q\neq r$) two constants $d$ and $e$ and a probability $p_{e}$ such that
    $\RV=\QV\cup\{q\}$, $\Rr=\Qr\otimes\dr{q\mapsto d}$,
    $\SV=\QV-\{r\}$ and $\Sr=\mis{r}{e}(\Qr)$. 
    As a consequence, starting from $D\equiv [\QV\cup\{q\}, \Qr\otimes\dr{q\mapsto d}, SL]$ and $E\equiv [\mis{r}{e}(\Qr), \QV-\{r\},  NR]$
    applying Lemma~\ref{lemma:uniformity} again, we obtain
    \begin{eqnarray*}
      D&\ppredto{\measr_r}{p_{e}}& [\mis{r}{{e}}(\Qr\otimes\dr{q\mapsto d}),\QV\cup\{q\}-\{r\},SR]\\
       &\equiv&[\mis{r}{{e}}(\Rr),\QV\cup\{q\}-\{r\},SR]\equiv F;\\
      E&\predto{\nw}& [(\mis{r}{e}(\Qr))\otimes\dr{q\mapsto d},\QV-\{r\}\cup\{q\},SR]\\
       &\equiv&[(\Sr)\otimes\dr{q\mapsto d},\SV\cup\{q\},SR]\equiv G.
    \end{eqnarray*}
    Clearly, $\QV\cup\{q\}-\{r\}\equiv \QV-\{r\}\cup\{q\}$. By Proposition~\ref{Proposition:MisProp}, case~\ref{primo}, 
    it is  possible to commute the measurement of the quantum variable $r$ with the creation of the quantum variable $q$, in fact they are distinct quantum variable.
    Then $\mis{r}{{e}}(\Qr\otimes\dr{q\mapsto d})$ and $(\mis{r}{e}(\Qr))\otimes\dr{q\mapsto d}$  give the same quantum register.
    We can conclude that $F\equiv G$.
  \item 
    If $\alpha=\measr_r,  \beta=\nw$, the case is symmetric to the previous one. 
  \item 
    If $\alpha=\measr_r, \beta\neq\nw,\measr_q$, then by Lemma~\ref{lemma:uniformity} there exist a constant $c$ and a probability $p_c$ such that 
    $\Rr=\mis{r}{c}(\Qr)$, $\RV=\QV-\{r\}$, $\SV=\QV$ and $\Sr=(U_{L,R}\otimes I_{\QV-{\Qvt{L}}})\Qr$. 
    As a consequence, starting from $D\equiv [\mis{r}{c}(\Qr), \QV-\{r\}, SL]$ and $E\equiv [(U_{L,R}\otimes I_{\QV-{\Qvt{L}}})\Qr, \QV, NR]$, 
    applying Lemma~\ref{lemma:uniformity} again, we obtain
    \begin{eqnarray*}
      D&\predto{\beta}&[(U_{L,R}\otimes I_{\QV-\{r\}-{\Qvt{L}}})(\mis{r}{c}(\Qr)),\QV-\{r\},SR]\\
       &\equiv& [(U_{L,R}\otimes I_{\QV-\{r\}-{\Qvt{L}}})(\Rr),\RV,SR]\equiv F\\
      E&\ppredto{\measr_r}{p_{c}}&[\mis{r}{c}((U_{L,R}\otimes I_{\QV-{\Qvt{L}}})\Qr),\QV-\{r\},SR]\\
       &\equiv& [\mis{r}{c}(\Sr),\QV-\{r\},SR]\equiv G
    \end{eqnarray*}
    Note that the operators $(U_{L,R}\otimes I_{\QV-\{r\}-{\Qvt{L}}})\circ \mis{r}{c}$ and $\mis{r}{c}\circ(U_{L,R}\otimes I_{\QV-{\Qvt{L}}})$ 
    act on $\Qr$ in the same way, by means of Proposition~\ref{Proposition:MisProp}, case~\ref{terzo}. 
    We can conclude that $F\equiv G$.
  \item 
    If $\alpha\neq\nw,\measr_q, \beta=\measr_r$, the case is symmetric to the previous one.
  \end{varitemize}
\item
  $D\equiv [\Rr,\RV,SL]$ and $E\equiv [\Sr,\SV,TL]$,
  where $[\Qr,QV,N]\redto[\Rr,\RV,S]$
  and $[\Qr,\QV,N]\redto[\Sr,\SV,T]$.
  Here we can apply the inductive hypothesis.
\item
  $D\equiv [\Rr,\RV,NR]$ and $E\equiv [\Sr,\SV,NU]$,
  where $[\Qr,QV,L]\redto[\Rr,\RV,R]$ and
  $[\Qr,\QV,L]\redto[\Sr,\SV,U]$.
  Here we can apply the inductive hypothesis as well.
\item 
  $N\equiv (\lambda x.P)$, $D\equiv [\Qr,\QV,P\{L/x\}]$, $E\equiv [\Rr,\RV,NR]$,
  where
  $[\Qr,\QV,L]\predto{\beta}[\Rr,\RV,R]$.
  Clearly $\Ok{[\Qr,\QV,P\{L/x\}]}$ and, by
  Lemma~\ref{lemma:paramuniformity},
  $[\Qr,\QV,P\{L/x\}]\redto[\Rr,\RV,P\{R/x\}]$.
  Moreover, $[\Rr,\RV,NR]\equiv [\Rr,\RV,(\lambda
  x.P)R]\redto[\Rr,\RV,P\{R/x\}]$.
\item
  \begin{sloppy}
  $N\equiv (\lambda x.P)$, $D\equiv [\Qr,\QV,P\{L/x\}]$,
  $E\equiv [\Rr,\RV,(\lambda x.V)L]$,
  where $[\Qr,\QV,P]\predto{\beta}[\Rr,\RV,V]$.
  Clearly $\Ok{[\Qr,\QV,P\{L/x\}]}$ and, by
  Lemma~\ref{lemma:paramuniformity},
  $[\Qr,\QV,P\{L/x\}]\predto{\beta}[\Rr,\RV,V\{L/x\}]$.
  Moreover, $[\Rr,\RV,(\lambda
  x.V)L]\predto{\beta}[\Rr,\RV,V\{L/x\}]$.
  \end{sloppy}
\item 
  \begin{sloppy}
  $N\equiv (\lambda !x.P)$, $L\equiv !Q$, $D\equiv
  [\Qr,\QV,P\{Q/x\}]$,
  $E\equiv [\Rr,\RV,(\lambda !x.V)L]$,
  where $[\Qr,\QV,P]\predto{\beta}[\Rr,\RV,V]$.
  Clearly $\Ok{[\Qr,\QV,P\{Q/x\}]}$ and, by
  Lemma~\ref{lemma:paramuniformity},
  $[\Qr,\QV,P\{Q/x\}]\predto{\beta}[\Rr,\RV,V\{Q/x\}]$.
  Moreover, $[\Rr,\RV,(\lambda
  x.V)!Q]\predto{\beta}[\Rr,\RV,V\{Q/x\}]$.
\item $N\equiv (\lambda <x_1,\ldots,x_n>.P)$, $L\equiv
  <r_1,\ldots,r_n>$,
  $D\equiv [\Qr,\QV,P\{r_1/x_1,\ldots,r_n/x_n\}]$,
  $E\equiv [\Rr,\RV,(\lambda <x_1,\ldots,x_n>.V)L]$,
  where $[\Qr,\QV,P]\predto{\beta}[\Rr,\RV,V]$.
  Clearly $\Ok{[\Qr,\QV,P\{r_1/x_1,\ldots,r_n/x_n\}]}$
  and, by Lemma~\ref{lemma:paramuniformity}, 
  $[\Qr,\QV,P\{r_1/x_1,\ldots,r_n/x_n\}]\predto{\beta}[\Rr,\RV,V\{r_1/x_1,\ldots,r_n/x_n\}]$.
  Moreover, $[\Rr,\RV,(\lambda
  <x_1,\ldots,x_n>.V)L]\predto{\beta}[\Rr,\RV,V\{r_1/x_1,\ldots,r_n/x_n\}]$.
\item
  $N\equiv (\lambda x.P)Q$, $D\equiv [\Qr,\QV,(\lambda x.PL)Q]$,
  $E\equiv [\Qr,\QV,(P\{Q/x\})L]$, $\alpha=\rcom$, $\beta=\lbeta$.
 Clearly, 
  $[\Qr,\QV,(\lambda x.PL)Q]\predto{\lbeta}[\Qr,\QV,(P\{Q/x\})L]$.
\item $N\equiv (\lambda\pi.P)Q$, $D\equiv [\Qr,\QV,(\lambda
  \pi.PL)Q]$,
  $E\equiv [\Rr,\RV,((\lambda\pi.V)Q)L]$, $\alpha=\rcom$,
  where $[\Qr,\QV,P]\predto{\beta}[\Rr,\RV,V]$.
  Clearly, $[\Qr,\QV,(\lambda x.PL)Q]\predto{\rcom}[\Rr,\RV,(\lambda
  x.VL)Q]$ and
  $[\Rr,\RV,((\lambda\pi.V)Q)L]\predto{\beta}[\Rr,\RV,(\lambda\pi.VL)Q]$.
\item $N\equiv (\lambda\pi.P)Q$, $D\equiv [\Qr,\QV,(\lambda x.PL)Q]$,
  $E\equiv [\Rr,\RV,((\lambda\pi.P)W)L]$, $\alpha=\rcom$,
  where $[\Qr,\QV,Q]\predto{\beta}[\Rr,\RV,W]$.
  Clearly, $[\Qr,\QV,(\lambda x.PL)Q]\predto{\rcom}[\Rr,\RV,(\lambda
  x.PL)W]$ and
  $[\Rr,\RV,((\lambda\pi.P)W)L]\predto{\beta}[\Rr,\RV,(\lambda\pi.PL)W]$.
\item $N\equiv (\lambda\pi.P)Q$, $D\equiv [\Qr,\QV,(\lambda x.PL)Q]$,
  $E\equiv [\Rr,\RV,((\lambda\pi.P)Q)R]$, $\alpha=\rcom$,
  where $[\Qr,\QV,L]\predto{\beta}[\Rr,\RV,R]$.
  Clearly, $[\Qr,\QV,(\lambda x.PL)Q]\predto{\rcom}[\Rr,\RV,(\lambda
  x.PR)Q]$ and
  $[\Rr,\RV,((\lambda\pi.P)Q)R]\predto{\beta}[\Rr,\RV,(\lambda\pi.PR)Q]$.
\item $N\equiv (\lambda\pi.P)$, $L\equiv (\lambda x.Q)R$, $D\equiv
  [\Qr,\QV,(\lambda x.NQ)R]$,
  $E\equiv [\Qr,\QV,N(Q\{R/x\})]$, $\alpha=\lcom$, $\beta=\lbeta$.
  Clearly, $[\Qr,\QV,(\lambda
  x.NQ)R]\redto{\lbeta}[\Qr,\QV,N(Q\{R/x\})]$.
  \end{sloppy}
\end{varitemize}
If $M$ is in the form $\new{c}$, then $D\equiv E$.
\end{proof}
\begin{remark}
Unfortunately, Proposition~\ref{prop:onestepconf} does
not translate into an equivalent result on mixed states, because
of commutative reduction rules. As a consequence,
it is more convenient to first study confluence at the level of 
probabilistic computations.
\end{remark}
Note that, even if the calculus is untyped,  we cannot build an infinite sequence of commuting reductions:
\begin{lemma}\label{lemma:noinfcom}
  The relation $\predto{\commrul}$ is strongly
  normalizing. In other words, there cannot be
  any infinite sequence $C_1\predto{\commrul} C_2\predto{\commrul} C_3\predto{\commrul}\ldots$. 
\end{lemma}
\begin{proof}
Define the size $|M|$ of a term $M$ as the number of symbols in it.
Moreover, define the abstraction size $|M|_\lambda$ of $M$ as the
sum over all subterms of $M$ in the form $\lambda\pi.N$, of
$|N|$. Clearly $|M|_\lambda\leq |M|^2$. Moreover, 
if $[\Qr,\QV,M]\predto{\commrul}[\Qr,\QV,N]$, then $|N|=|M|$ but $|N|_\lambda\gt|M|_\lambda$.
This concludes the proof.
\end{proof}}{
\qstar\ enjoys a form of \emph{quasi-one-step confluence}. As an example, if 
$\confone\predto{\noncommrul}\conftwo$ and $\confone\predto{\noncommrul}\confthree$
then there is $\conffour$ with $\conftwo\predto{\noncommrul}\conffour$ and 
$\confthree\predto{\noncommrul}\conffour$. If, on the other hand,
$\confone\predto{\noncommrul}\conftwo$ and $\confone\predto{\commrul}\confthree$
then either $\confthree\predto{\noncommrul}\conftwo$ or there is $\conffour$ as above.
As another interesting example, if $\confone\predto{\noncommrul}\conftwo$
and $\confone\predto{\measr_r}\confthree$ then there is $\conffour$
as above. The only problematic case is when
$\confone\predto{\measr_r}\conftwo$ and $\confone\predto{\measr_r}\confthree$, which cannot
be solved. Lack of space prevents us from formally stating and proving quasi-one-step confluence,
which can anyway be found in~\cite{ExtendedVersion}. Quasi-one-step confluence is an
essential ingredient towards strong confluence.
}


We define the \emph{weight} $\weight{\pcompone}$ and the \emph{branch degree} $\bdegree{\pcompone}$ of
every \emph{finite} probabilistic computation $\pcompone$ by induction on the structure of $\pcompone$:
\begin{varitemize}
\item
  $\weight{[\confone]}=0$ and $\bdegree{[\confone]}=1$.
\item
  $\bdegree{[\confone,\pcompone]}=\bdegree{\pcompone}$. 
  Moreover, let $\conftwo$ be the root of $\pcompone$. If 
  $\confone\redto_{\commrul}\conftwo$, then
  $\weight{[\confone,\pcompone]}=\weight{\pcompone}$, otherwise
  $\weight{[\confone,\pcompone]}=\bdegree{\pcompone}+\weight{\pcompone}$.
\item
  $\bdegree{[(p,\confone),\pcompone,\pcomptwo]}=\bdegree{\pcompone}+\bdegree{\pcomptwo}$, while
  $\weight{[(p,\confone),\pcompone,\pcomptwo]}=\bdegree{\pcompone}+\bdegree{\pcomptwo}
  +\weight{\pcompone}+\weight{\pcomptwo}$.
\end{varitemize}
Please observe that $\bdegree{\pcompone}\geq 1$ for every $\pcompone$.

Now we propose a robabilistic variation on the classical \textit{strip lemma} of the $\lambda$-calculus. 
It will have a crucial  r\^ole in the proof of strong confluence (Theorem~\ref{theor:same-distr}).

\begin{lemma}[Probabilistic Strip Lemma]\label{prop:pcompconf}
Let $\pcompone$ be a finite probabilistic computation with root $\confone$
and positive weight $\weight{\pcompone}$.
\begin{varitemize}
\item
  If $\confone\redto_{\noncommrul}\conftwo$, then
  there is $\pcomptwo$ with root $\conftwo$ such that
  $\weight{\pcomptwo}\lt\weight{\pcompone}$,
  $\bdegree{\pcomptwo}\leq\bdegree{\pcompone}$
  and for every $\confthree\in\NF$, it holds that
  $\probmc{\pcomptwo}{\confthree}\geq\probmc{\pcompone}{\confthree}$,
  $\numlv{\pcomptwo}{\confthree}\geq\numlv{\pcompone}{\confthree}$,
  $\probmcany{\pcomptwo}\geq\probmcany{\pcompone}$ and
  $\numlvany{\pcomptwo}\geq\numlvany{\pcompone}$.
\item
  If $\confone\redto_{\commrul}\conftwo$, then
  there is $\pcomptwo$ with root $\conftwo$ such that
  $\weight{\pcomptwo}\leq\weight{\pcompone}$,
  $\bdegree{\pcomptwo}\leq\bdegree{\pcompone}$
  and for every $\confthree\in\NF$, it holds that
  $\probmc{\pcomptwo}{\confthree}\geq\probmc{\pcompone}{\confthree}$,
  $\numlv{\pcomptwo}{\confthree}\geq\numlv{\pcompone}{\confthree}$,
  $\probmcany{\pcomptwo}\geq\probmcany{\pcompone}$ and
  $\numlvany{\pcomptwo}\geq\numlvany{\pcompone}$.
\item
  If $\confone\ppredto{\measr_r}{q}\conftwo$
  and $\confone\ppredto{\measr_r}{p}\confthree$, then
  there are $\pcomptwo$ and $\pcompthree$ with roots
  $\conftwo$ and $\confthree$ such that
  $\weight{\pcomptwo}\lt\weight{\pcompone}$,
  $\weight{\pcompthree}\lt\weight{\pcompone}$,
  $\bdegree{\pcomptwo}\leq\bdegree{\pcompone}$,
  $\bdegree{\pcompthree}\leq\bdegree{\pcompone}$
  and for every $\confthree\in\NF$, it holds that
  $q\probmc{\pcomptwo}{\confthree}+p\probmc{\pcompthree}{\confthree}\geq\probmc{\pcompone}{\confthree}$,
  $\numlv{\pcomptwo}{\confthree}+\numlv{\pcompthree}{\confthree}\geq\numlv{\pcompone}{\confthree}$,
  $q\probmcany{\pcomptwo}+p\probmcany{\pcompthree}\geq\probmcany{\pcompone}$ and
  $\numlvany{\pcomptwo}+\numlvany{\pcompthree}\geq\numlvany{\pcompone}$.
\end{varitemize}
\end{lemma}
\condinc{
\begin{proof}
By induction on the structure of $\pcompone$:
\begin{varitemize}
\item
  $\pcompone$ cannot simply be $[\confone]$, because
  $\weight{\pcompone}\geq 1$.
\item
  If $\pcompone=[\confone,\pcompfour]$, where $\pcompfour$ has root
  $\conffour$ and $\confone\redto_{\noncommrul}\conffour$, then:
  \begin{varitemize}
  \item Suppose $\confone\redto_{\noncommrul}\conftwo$. If
    $\conftwo=\conffour$, then the required $\pcomptwo$ is simply
    $\pcompfour$. Otherwise, by Proposition $\ref{prop:onestepconf}$,
    there is $\conffive$ such that
    $\conftwo\redto_{\noncommrul}\conffive$ and
    $\conffour\redto_{\noncommrul}\conffive$. Now, if $\pcompfour$ is
    simply $[\conffour]$, then the required probabilistic computation is simply
    $[\conftwo]$, because neither $\conffour$ nor $\conftwo$ are in
    normal form and, moreover, $\weight{[\conftwo]}=0\lt
    1=\weight{\pcompone}$.  If, on the other hand, $\pcompfour$ has
    positive weight we can apply the IH to it, obtaining a
    probabilistic computation $\pcompfive$ with root $\conffive$ such that
    $\weight{\pcompfive}\lt\weight{\pcompfour}$,
    $\bdegree{\pcompfive}\leq\bdegree{\pcompfour}$,
    $\probmc{\pcompfive}{\confsix}\geq\probmc{\pcompfour}{\confsix}$ and
    $\numlv{\pcompfive}{\confsix}\geq\numlv{\pcompfour}{\confsix}$
    for every $\confsix\in\NF$. Then, the required probabilistic computation is
    $[\conftwo,\pcompfive]$, since
    \begin{eqnarray*}
      \weight{[\conftwo,\pcompfive]}&=&\bdegree{\pcompfive}+\weight{\pcompfive}\lt\bdegree{\pcompfive}+\weight{\pcompfour}\\
      &\leq&\bdegree{\pcompfour}+\weight{\pcompfour}=\weight{\pcompone};\\
      \probmc{[\conftwo,\pcompfive]}{\confsix}&=&\probmc{\pcompfive}{\confsix}\geq\probmc{\pcompfour}{\confsix}\\
      &=&\probmc{\pcompone}{\confsix};\\
      \numlv{[\conftwo,\pcompfive]}{\confsix}&=&\numlv{\pcompfive}{\confsix}\geq\numlv{\pcompfour}{\confsix}\\
      &=&\numlv{\pcompone}{\confsix}.
    \end{eqnarray*}
  \item
    Suppose $\confone\redto_{\commrul}\conftwo$. By Proposition
    \ref{prop:onestepconf} one of the following two cases applies:
    \begin{varitemize}
    \item There is $\conffive$ such that
      $\conftwo\redto_{\noncommrul}\conffive$ and
      $\conffour\redto_{\commrul}\conffive$ Now, if $\pcompfour$ is
      simply $[\conffour]$, then the required probabilistic computation is
      simply $[\conftwo,[\conffive]]$, because
      $\weight{[\conftwo,[\conffive]]}=1=\weight{\pcompone}$.  If, on
      the other hand, $\pcompfour$ has positive weight we can apply
      the IH to it, obtaining a probabilistic computation $\pcompfive$ with root
      $\conffive$ such that
      $\weight{\pcompfive}\leq\weight{\pcompfour}$,
      $\bdegree{\pcompfive}\leq\bdegree{\pcompfour}$ and
      $\probmc{\pcompfive}{\confsix}\geq\probmc{\pcompfive}{\confsix}$
      for every $\confsix\in\NF$. Then, the required probabilistic computation is
      $[\conftwo,\pcompfive]$, since
      \begin{eqnarray*}
        \weight{[\conftwo,\pcompfive]}&=&\bdegree{\pcompfive}+\weight{\pcompfive}\leq\bdegree{\pcompfive}+\weight{\pcompfour}\\
        &\leq&\bdegree{\pcompfour}+\weight{\pcompfour}=\weight{\pcompone}\\
        \probmc{[\conftwo,\pcompfive]}{\confsix}&=&\probmc{\pcompfive}{\confsix}\geq\probmc{\pcompfour}{\confsix}\\
        &=&\probmc{\pcompone}{\confsix};\\
        \numlv{[\conftwo,\pcompfive]}{\confsix}&=&\numlv{\pcompfive}{\confsix}\geq\numlv{\pcompfour}{\confsix}\\
        &=&\numlv{\pcompone}{\confsix}.
      \end{eqnarray*}
    \item
      $\conftwo\redto_{\noncommrul}\conffour$. The required probabilistic computation is simply $[\conftwo,\pcompfour]$. Indeed:
      \begin{eqnarray*}
        \weight{[\conftwo,\pcompfour]}&=&\bdegree{\pcompfour}+\weight{\pcompfour}=\weight{[\confone,\pcompfour]}=\weight{[\pcompone]}.
      \end{eqnarray*}
    \end{varitemize}
  \item Suppose $\confone\ppredto{\measr_r}{q}\conftwo$ 
    and $\confone\ppredto{\measr_r}{p}\confthree$. 
    By Proposition~\ref{prop:onestepconf},
    there are $\conffive$ and $\confsix$ such that
    $\conftwo\redto_{\noncommrul}\conffive$,
    $\confthree\redto_{\noncommrul}\confsix$,
    $\conffour\ppredto{\measr_r}{q}\conffive$,
    $\conffour\ppredto{\measr_r}{p}\confsix$.  
    Now, if $\pcompfour$ is simply $\conffour$, then the
    required probabilistic computations are simply $[\conftwo]$ and
    $[\confthree]$, because neither $\conffour$ nor $\conftwo$
    nor $\confthree$ are in normal form and, moreover,
    $\weight{[\conftwo]}=\weight{[\confthree]}=0\lt
    1=\weight{\pcompone}$.  If, on the other hand, $\pcompfour$ has
    positive weight we can apply the IH to it, obtaining
    probabilistic computations $\pcompfive$ and $\pcompsix$ with roots
    $\conffive$ and $\confsix$ such that
    $\weight{\pcompfive}\lt\weight{\pcompfour}$,
    $\weight{\pcompsix}\lt\weight{\pcompfour}$,
    $\bdegree{\pcompfive}\leq\bdegree{\pcompfour}$,
    $\bdegree{\pcompsix}\leq\bdegree{\pcompfour}$,
    $q\probmc{\pcompfive}{\confsix}+(p)\probmc{\pcompsix}{\confsix}\geq\probmc{\pcompfour}{\confsix}$ and
    $\numlv{\pcompfive}{\confsix}+\numlv{\pcompsix}{\confsix}\geq\numlv{\pcompfour}{\confsix}$
    for every $\confsix\in\NF$. Then, the required probabilistic computations are
    $[\conftwo,\pcompfive]$ and $[\confthree,\pcompsix]$, since
    \begin{eqnarray*}
      \weight{[\conftwo,\pcompfive]}&=&\bdegree{\pcompfive}+\weight{\pcompfive}\lt\bdegree{\pcompfive}+\weight{\pcompfour}\\
      &\leq&\bdegree{\pcompfour}+\weight{\pcompfour}=\weight{\pcompone};\\
      \weight{[\confthree,\pcompsix]}&=&\bdegree{\pcompsix}+\weight{\pcompsix}\lt\bdegree{\pcompsix}+\weight{\pcompfour}\\
      &\leq&\bdegree{\pcompfour}+\weight{\pcompfour}=\weight{\pcompone}.
    \end{eqnarray*}
    Moreover, for every $\confsix\in\NF$
    \begin{eqnarray*}
      q\probmc{[\conftwo,\pcompfive]}{\confsix}+p\probmc{[\confthree,\pcompsix]}{\confsix}&=&
      q\probmc{\pcompfive}{\confsix}+p\probmc{\pcompsix}{\confsix}\\
      &\geq&\probmc{\pcompfour}{\confsix}=\probmc{\pcompone}{\confsix}\\
      \numlv{[\conftwo,\pcompfive]}{\confsix}+\numlv{[\confthree,\pcompsix]}{\confsix}&=&
      \numlv{\pcompfive}{\confsix}+\numlv{\pcompsix}{\confsix}\\
      &\geq&\numlv{\pcompfour}{\confsix}=\numlv{\pcompone}{\confsix}
    \end{eqnarray*}
  \end{varitemize}
  \item
    The other cases are similar.
\end{varitemize}
\end{proof}}{}
The following Proposition follows from the probabilistic strip lemma. It can be read as a simulation result:
if $\pcompone$ and $\pcomptwo$ are maximal and have the same root, then $\pcompone$ can simulate $\pcomptwo$
(and viceversa).
\begin{proposition}\label{prop:infpcompconf}
  For every maximal probabilistic computations $\pcompone$ and for every finite probabilistic computation $\pcomptwo$ such that $\pcompone$ 
  and $\pcomptwo$ have the same root, there
  is a finite sub-computation $\pcompthree$ of $\pcompone$ such that for every $\confone\in\NF$,
  $\probmc{\pcompthree}{\confone}\geq\probmc{\pcomptwo}{\confone}$ and
  $\numlv{\pcompthree}{\confone}\geq\numlv{\pcomptwo}{\confone}$.
  Moreover,
  $\probmcany{\pcompthree}\geq\probmcany{\pcomptwo}$ and
  $\numlvany{\pcompthree}\geq\numlvany{\pcomptwo}$.
\end{proposition}
\condinc{
\begin{proof}
Given any probabilistic computation $\pcompfour$, its $\commrul$-degree $n_\pcompfour$ is the number of consecutive commutative
rules you find descending $\pcompfour$, starting at the root. By Lemma~\ref{lemma:noinfcom}, this is a good
definition. The proof goes by induction on $(\weight{\pcomptwo},n_\pcomptwo)$, ordered lexicographically: 
\begin{varitemize}
\item
  If $\weight{\pcomptwo}=0$, then $\pcomptwo$ is just $[\conftwo]$ for some configuration $\conftwo$.
  Then, $\pcompthree=\pcomptwo$ and all the required conditions hold.
\item
  If $\weight{\pcomptwo}\gt 0$, then we distinguish three cases, depending on the shape of $\pcompone$:
  \begin{varitemize}
  \item
    If $\pcompone=[\conftwo,\pcompfour]$, $\confthree$ is the root of $\pcompfour$ and $\conftwo\redto_{\noncommrul}\confthree$,
    then, by Proposition~\ref{prop:pcompconf}, there is a probabilistic computation $\pcompfive$ with root $\confthree$
    such that $\weight{\pcompfive}\lt\weight{\pcomptwo}$
    and $\probmc{\pcompfive}{\confone}\geq\probmc{\pcomptwo}{\confone}$ for every $\confone\in\NF$. By the inductive hypothesis applied
    to $\pcompfour$ and $\pcompfive$, there is a sub-probabilistic computation $\pcompsix$ of $\pcompfour$ such that 
    $\probmc{\pcompsix}{\confone}\geq\probmc{\pcompfive}{\confone}$ and
    $\numlv{\pcompsix}{\confone}\geq\numlv{\pcompfive}{\confone}$ for every $\confone\in\NF$. Now, consider the
    probabilistic computation $[\conftwo,\pcompsix]$. This is clearly a sub-probabilistic computation of $\pcompone$. Moreover,
    for every $\confone\in\NF$:
    \begin{eqnarray*}
      \probmc{[\conftwo,\pcompsix]}{\confone}&=&\probmc{\pcompsix}{\confone}\\
      &\geq&\probmc{\pcompfive}{\confone}\geq\probmc{\pcomptwo}{\confone}\\
      \numlv{[\conftwo,\pcompsix]}{\confone}&=&\numlv{\pcompsix}{\confone}\\
      &\geq&\numlv{\pcompfive}{\confone}\geq\numlv{\pcomptwo}{\confone}.
    \end{eqnarray*}
  \item
    If $\pcompone=[\conftwo,\pcompfour]$, $\confthree$ is the root of $\pcompfour$ and $\conftwo\redto_{\commrul}\confthree$,
    then, by Proposition~\ref{prop:pcompconf}, there is a probabilistic computation $\pcompfive$ with root $\confthree$
    such that $\weight{\pcompfive}\leq\weight{\pcomptwo}$
    and $\probmc{\pcompfive}{\confone}\geq\probmc{\pcomptwo}{\confone}$ for every $\confone\in\NF$. Now, observe we can
    apply the inductive hypothesis to $\pcompfour$ and $\pcompfive$, because $\weight{\pcompfive}\leq\weight{\pcomptwo}$
    and $n_\pcompfour\lt n_\pcompone$. So, there is a sub-probabilistic computation $\pcompsix$ of $\pcompfour$ such that 
    $\probmc{\pcompsix}{\confone}\geq\probmc{\pcompfive}{\confone}$ and
    $\numlv{\pcompsix}{\confone}\geq\numlv{\pcompfive}{\confone}$ 
    for every $\confone\in\NF$. Now, consider the
    probabilistic computation $[\conftwo,\pcompsix]$. This is clearly a sub-probabilistic computation of $\pcompone$. Moreover,
    for every $\confone\in\NF$:
    \begin{eqnarray*}
      \probmc{[\conftwo,\pcompsix]}{\confone}&=&\probmc{\pcompsix}{\confone}\\
      &\geq&\probmc{\pcompfive}{\confone}\geq\probmc{\pcomptwo}{\confone}\\
      \numlv{[\conftwo,\pcompsix]}{\confone}&=&\numlv{\pcompsix}{\confone}\\
      &\geq&\numlv{\pcompfive}{\confone}\geq\numlv{\pcomptwo}{\confone}.
    \end{eqnarray*}
  \item
    $\pcompone=[(p,q,\conftwo),\pcompfour_1,\pcompfour_2]$, $\confthree_1$ is the root of $\pcompfour_1$ and
    $\confthree_2$ is the root of $\pcompfour_2$, then, by Proposition~\ref{prop:pcompconf}, there are
    probabilistic computations $\pcompfive_1$ and $\pcompfive_2$ with root $\confthree_1$ and $\confthree_2$
    such that $\weight{\pcompfive_1},\weight{\pcompfive_2}\lt\weight{\pcomptwo}$
    and $p\probmc{\pcompfive_1}{\confone}+q\probmc{\pcompfive_2}{\confone}\geq\probmc{\pcomptwo}{\confone}$ 
    for every $\confone\in\NF$. By the inductive hypothesis applied
    to $\pcompfour_1$ and $\pcompfive_1$ (to $\pcompfour_2$ and $\pcompfive_2$, respectively), 
    there is a sub-probabilistic computation $\pcompsix_1$ of $\pcompfour_1$ (a sub-probabilistic computation $\pcompsix_2$ of $\pcompfour_2$, respectively) 
    such that $\probmc{\pcompsix_1}{\confone}\geq\probmc{\pcompfive_1}{\confone}$ for every $\confone\in\NF$
    ($\probmc{\pcompsix_2}{\confone}\geq\probmc{\pcompfive_2}{\confone}$ for every $\confone\in\NF$, respectively).
    Now, consider the probabilistic computation $[(p,q,\conftwo),\pcompsix_1,\pcompsix_2]$. This is clearly a sub-probabilistic computation of $\pcompone$. Moreover,
    for every $\confone\in\NF$:
    \begin{eqnarray*}
      \probmc{[(p,q,\conftwo,\pcompsix_1,\pcompsix_2]}{\confone}&=&p\probmc{\pcompsix_1}{\confone}+q\probmc{\pcompsix_2}{\confone}\\
      &\geq&p\probmc{\pcompfive_1}{\confone}+q\probmc{\pcompfive_2}{\confone}\geq\probmc{\pcomptwo}{\confone}.
    \end{eqnarray*}    
  \end{varitemize}
\end{varitemize}
This concludes the proof.
\end{proof}}{}
The main theorem is the following:
\begin{theorem}[Strong Confluence]\label{theor:same-distr}
  For every maximal  probabilistic computation $\pcompone$, for every maximal probabilistic computation $\pcomptwo$ such that $\pcompone$ 
  and $\pcomptwo$ have the same root, and for every $\confone\in\NF$, $\probmc{\pcompone}{\confone} = \probmc{\pcomptwo}{\confone}$
  and $\numlv{\pcompone}{\confone}=\numlv{\pcomptwo}{\confone}$. Moreover, $\probmcany{\pcompone}=\probmcany{\pcomptwo}$
  and $\numlvany{\pcompone}=\numlvany{\pcomptwo}$.
\end{theorem}
\condinc{
\begin{proof}
  Let $\confone\in\NF$ be any configuration in normal form.
  Clearly:
  $$
  \begin{array}{c@{\qquad}c}
    \probmc{\pcompone}{\confone}=\sup_{\pcompthree\sqsubseteq\pcompone}\{\probmc{\pcompthree}{\confone}\}&
    \probmc{\pcomptwo}{\confone}=\sup_{\pcompfour\sqsubseteq\pcomptwo}\{\probmc{\pcompfour}{\confone}\}
  \end{array}
  $$
  Now, consider the two sets $A=\{\probmc{\pcompthree}{\confone}\}_{\pcompthree\sqsubseteq\pcompone}$ and 
  $B=\{\probmc{\pcompfour}{\confone}\}_{\pcompfour\sqsubseteq\pcomptwo}$. We claim the two sets have the
  same upper bounds. Indeed, if $x\in\mathbb{R}$ is an upper bound
  on $A$ and $\pcompfour\sqsubseteq\pcomptwo$, by Proposition~\ref{prop:infpcompconf}
  there is $\pcompthree\sqsubseteq\pcompone$ such
  that $\probmc{\pcompthree}{\confone}\geq\probmc{\pcompfour}{\confone}$,
  and so $x\geq\probmc{\pcompfour}{\confone}$. As a consequence, $x$ is an upper bound on $B$.
  Symmetrically, if $x$ is an upper bound on $B$, it is an upper bound on $A$.
  Since $A$ and $B$ have the same upper bounds, they have the same least upper bound, and
  $\probmc{\pcompone}{\confone}=\probmc{\pcomptwo}{\confone}$.
  The other claims can be proved exactly in the same way.
  This concludes the proof.
\end{proof}}{}


\section{Computing with Mixed States}\label{sec:MixS} 

\condinc{\begin{definition}[Mixed State]}{\begin{definition}}
  A mixed state is a function $\mixone: \conf \redto
  \mathbb{R}_{[0,1]}$ such that there is a finite set $S\subseteq\conf$
  with $\mixone(\confone)=0$ except when $\confone\in S$ and, moreover,
  $\sum_{\confone\in S} \mixone(\confone)=1$.  
  $\Mix$\ is the set of mixed states.
\end{definition}
In this paper, a mixed state $\mixone$ will be denoted with the linear notation
$
\{ p_1:\confone_1,\ldots,p_k:\confone_k\}
$
or as $\{p_i:\confone_i\}_{1\leq i\leq k}$,
where $p_i$ is the probability $\mixone(\confone_i)$ associated to the configuration $\confone_i$.

\condinc{\begin{definition}[Reduction]}{\begin{definition}}
  The reduction relation $\parred$ between mixed states is defined
  in the following way:
  $ \{ p_1:\confone_1,\ldots,p_m:\confone_m\}\parred \mixone $ \textit{iff} 
  there exist $m$ mixed states  $\mixone_1=\{q^i_1:D^i_1\}_{1\leq i\in n_1},\ldots, 
  \mixone_m=\{q^i_m:D^i_m\}_{1\leq i\leq n_m}$ such that:
  \begin{varenumerate}
  \item  
    For every $i\in[1,m]$, it holds that $1\leq n_i\leq 2$;
  \item
    If $n_i=1$, then 
    either
    $\confone_i$ is in normal form and $\confone_i = \conftwo^1_i$ 
    or $\confone_i\redto_{\nonmeasrul} \conftwo^1_i$; 
  \item
    If $n_i=2$, then
    $\confone_i\ppredto{\measr_r}{p}\conftwo^1_i$, 
    $\confone_i\ppredto{\measr_r}{q}\conftwo^2_i$, $p\in\mathbb{R}_{[0,1]}$, 
    and $q^1_i= p, q^2_k= q$;
  \item 
  $\forall D\in\conf.\ \mixone(D)=\sum_{i=1}^{m} p_i\cdot\mixone_i(D)$.
\end{varenumerate}
\end{definition}
Given the reduction relation $\parred$, the corresponding notion of computation 
(that we call \textit{mixed computation}, in order to emphasize that mixed
states play the role of configurations) is completely standard. 

Given a mixed state $\mixone$ and a configuration $\confone\in\NF$, the
\emph{probability} of observing $\confone$ in $\mixone$ is
defined as $\mixone(C)$ and is denoted as $\probmc{\mixone}{\confone}$.
Observe that if $\mixone\parred\mixtwo$ and $\confone\in\NF$, then
$\probmc{\mixone}{\confone}\leq\probmc{\mixtwo}{\confone}$.
If $\{\mixone_i\}_{i\lt\varphi}$ is a mixed computation, then
$$
\sup_{i\lt\varphi}\probmc{\mixone_i}{\confone}
$$
always exists, and is denoted as $\probmc{\{\mixone_i\}_{i\lt\varphi}}{\confone}$. 

Please notice that a maximal mixed computation is always infinite. Indeed,
if $\mixone=\{p_i:\confone_i\}_{1\leq i\leq n}$ and for every $i\in[1,n], \confone_i\in\NF$, then
$\mixone\parred\mixone$.

\begin{proposition}\label{prop:infinitemixed}
   Let $\{\mixone_i\}_{i\lt\omega}$ be a maximal mixed computation and let $\confone_1,\ldots,\confone_n$ be 
   the configurations on which $\mixone_0$ evaluates to 
   a positive real. Then there are maximal probabilistic computations
   $\pcompone_1,\ldots,\pcompone_n$ with roots $\confone_1,\ldots,\confone_n$ such that
   $\sup_{j\lt\varphi}\mixone_j(\conftwo) = \sum_{i=1}^n\left(\mixone_0(\confone_i)\probmc{\pcompone_i}{\conftwo}\right)$
   for every $\conftwo$.
\end{proposition}
\condinc{
\begin{proof}
Let $\{\mixone_i\}_{i\lt\omega}$ be a maximal mixed computation.
Observe that $\mixone_0\parred^m\mixone_m$ for every $m\in\NN$. 
For every $m\in\NN$ let $\mixone_m$ be
$$
\{p_1^m:\confone_1^m,\ldots,p_{n_m}^m:\confone_{n_m}^m\}
$$
For every $m$, we can build maximal probabilistic computations
$\pcompone_{1}^m,\ldots,\pcompone_{n_m}^m$, generatively: assuming
$\pcompone_{1}^{m+1},\ldots,\pcompone_{n_{m+1}}^{m+1}$ are the
probabilistic computations corresponding to $\{\mixone_i\}_{m+1\leq
  i\lt\omega}$, they can be extended (and possibly merged) into some
maximal probabilistic computations
$\pcompone_{1}^m,\ldots,\pcompone_{n_m}^m$ corresponding to
$\{\mixone_i\}_{m\leq i\lt\omega}$. But we can even define for every
$m,k\in\NN$ with $m\leq k$, some finite probabilistic computations
$\pcompthree_1^{m,k},\ldots,\pcompthree_{n_{m}}^{m,k}$ with root
$\confone_1,\ldots,\confone_{n_m}$ and such that, for every $m,k$,
\begin{eqnarray*}
\pcompthree_i^{m,k}&\sqsubseteq&\pcompone_i^m\\
\mixone_k(\conftwo)&=&\sum_{i=1}^{n_m}\left(\mixone_m(\confone_i)\probmc{\pcompthree_i^{m,k}}{\conftwo}\right).
\end{eqnarray*}
This proceeds by induction on $k-m$. We can easily prove that for
every $\pcompfour\sqsubseteq\pcompone_i^m$ there is $k$ such that
$\pcompfour\sqsubseteq\pcompthree_i^{m,k}$: this is an induction on
$\pcompfour$ (which is a finite probabilistic computation). But now,
for every $\conftwo\in\NF$,
\begin{eqnarray*}
\sup_{j\lt\omega}\mixone_j(\conftwo)&=&\sup_{j\lt\omega}\sum_{i=1}^{n_0}\left(\mixone_0(\confone_i)\probmc{\pcompthree_i^{0,j}}{\conftwo}\right)\\
&=&\sum_{i=1}^{n_0}\left(\mixone_0(\confone_i)\sup_{j\lt\omega}\probmc{\pcompthree_i^{0,j}}{\conftwo}\right)\\
&=&\sum_{i=1}^{n_0}\left(\mixone_0(\confone_i)\probmc{\pcompone_i^0}{\conftwo}\right)\\
\end{eqnarray*}
This concludes the proof.
\end{proof}}{}

\begin{theorem}
For any two maximal mixed computations $\{\mixone_i\}_{i\lt\omega}$ and
$\{\mixtwo_i\}_{i\lt\omega}$ such that $\mixone_0=\mixtwo_0$, the following condition holds:
for every $\confone\in\NF$,
$\probmc{\{\mixone_i\}_{i\lt\omega}}{\confone}=\probmc{\{\mixtwo_i\}_{i\lt\omega}}{\confone}$
\end{theorem}
\condinc{
\begin{proof}
A trivial consequence of Proposition~\ref{prop:infinitemixed}.
\end{proof}}{}
\condinc{
\section{Conclusions}
The quantum lambda calculus \qstar\ is proved to enjoy confluence in a very strong form, both for finite and for
infinite computations. The proof seems to be quite independent on the particular rewriting system under consideration.
Actually, the authors believe that any rewriting system enjoying properties like Proposition~\ref{prop:onestepconf} enjoys
confluence in the same sense as the one used here. Indeed, this constitutes an interesting topic
for further work, which anyway lies outside the scope of this paper.}{}
\bibliography{biblio} 
\bibliographystyle{abbrv}
\end{document}

%% file: global.tex
\usepackage{amsmath,amssymb}
\usepackage{xspace}
\usepackage{stmaryrd}
\usepackage{url}
\usepackage{epsf}
\usepackage{graphics}

\usepackage{xy}
\xyoption{frame}
\xyoption{arrow}
\xyoption{curve}
\xyoption{arc}
\xyoption{line}
\xyoption{matrix}
\xyoption{tile}
\xyoption{ps}

\condinc{
    \newtheorem{lemma}{Lemma}
}

\newcount\mscount
\input{prooftree}

\newcommand{\dlinea}{\leavevmode\hrule\vspace{1pt}\hrule\mbox{}}

%% file: prooftree.tex
\newdimen\proofrulebreadth \proofrulebreadth=.05em
\newdimen\proofdotseparation \proofdotseparation=1.25ex
\newdimen\proofrulebaseline \proofrulebaseline=2ex
\newcount\proofdotnumber \proofdotnumber=3
\let\then\relax
\def\hfi{\hskip0pt plus.0001fil}
\mathchardef\squigto="3A3B
\newif\ifinsideprooftree\insideprooftreefalse
\newif\ifonleftofproofrule\onleftofproofrulefalse
\newif\ifproofdots\proofdotsfalse
\newif\ifdoubleproof\doubleprooffalse
\let\wereinproofbit\relax
\newdimen\shortenproofleft
\newdimen\shortenproofright
\newdimen\proofbelowshift
\newbox\proofabove
\newbox\proofbelow
\newbox\proofrulename
\def\shiftproofbelow{\let\next\relax\afterassignment\setshiftproofbelow\dimen0 }
\def\shiftproofbelowneg{\def\next{\multiply\dimen0 by-1 }%
\afterassignment\setshiftproofbelow\dimen0 }
\def\setshiftproofbelow{\next\proofbelowshift=\dimen0 }
\def\setproofrulebreadth{\proofrulebreadth}

\def\prooftree{
\ifnum  \lastpenalty=1
\then   \unpenalty
\else   \onleftofproofrulefalse
\fi
\ifonleftofproofrule
\else   \ifinsideprooftree
        \then   \hskip.5em plus1fil
        \fi
\fi
\bgroup
\setbox\proofbelow=\hbox{}\setbox\proofrulename=\hbox{}%
\let\justifies\proofover\let\leadsto\proofoverdots\let\Justifies\proofoverdbl
\let\using\proofusing\let\[\prooftree
\ifinsideprooftree\let\]\endprooftree\fi
\proofdotsfalse\doubleprooffalse
\let\thickness\setproofrulebreadth
\let\shiftright\shiftproofbelow \let\shift\shiftproofbelow
\let\shiftleft\shiftproofbelowneg
\let\ifwasinsideprooftree\ifinsideprooftree
\insideprooftreetrue
\setbox\proofabove=\hbox\bgroup$\displaystyle 
\let\wereinproofbit\prooftree
\shortenproofleft=0pt \shortenproofright=0pt \proofbelowshift=0pt
\onleftofproofruletrue\penalty1
}

\def\eproofbit{
\ifx    \wereinproofbit\prooftree
\then   \ifcase \lastpenalty
        \then   \shortenproofright=0pt  
        \or     \unpenalty\hfil         
        \or     \unpenalty\unskip       
        \else   \shortenproofright=0pt  
        \fi
\fi
\global\dimen0=\shortenproofleft
\global\dimen1=\shortenproofright
\global\dimen2=\proofrulebreadth
\global\dimen3=\proofbelowshift
\global\dimen4=\proofdotseparation
\global\count255=\proofdotnumber
$\egroup  
\shortenproofleft=\dimen0
\shortenproofright=\dimen1
\proofrulebreadth=\dimen2
\proofbelowshift=\dimen3
\proofdotseparation=\dimen4
\proofdotnumber=\count255
}

\def\proofover{
\eproofbit 
\setbox\proofbelow=\hbox\bgroup 
\let\wereinproofbit\proofover
$\displaystyle
}%
\def\proofoverdbl{
\eproofbit 
\doubleprooftrue
\setbox\proofbelow=\hbox\bgroup 
\let\wereinproofbit\proofoverdbl
$\displaystyle
}%
\def\proofoverdots{
\eproofbit 
\proofdotstrue
\setbox\proofbelow=\hbox\bgroup 
\let\wereinproofbit\proofoverdots
$\displaystyle
}%
\def\proofusing{
\eproofbit 
\setbox\proofrulename=\hbox\bgroup 
\let\wereinproofbit\proofusing
\kern0.3em$
}

\def\endprooftree{
\eproofbit 
  \dimen5 =0pt
\dimen0=\wd\proofabove \advance\dimen0-\shortenproofleft
\advance\dimen0-\shortenproofright
\dimen1=.5\dimen0 \advance\dimen1-.5\wd\proofbelow
\dimen4=\dimen1
\advance\dimen1\proofbelowshift \advance\dimen4-\proofbelowshift
\ifdim  \dimen1<0pt
\then   \advance\shortenproofleft\dimen1
        \advance\dimen0-\dimen1
        \dimen1=0pt
        \ifdim  \shortenproofleft<0pt
        \then   \setbox\proofabove=\hbox{%
                        \kern-\shortenproofleft\unhbox\proofabove}%
                \shortenproofleft=0pt
        \fi
\fi
\ifdim  \dimen4<0pt
\then   \advance\shortenproofright\dimen4
        \advance\dimen0-\dimen4
        \dimen4=0pt
\fi
\ifdim  \shortenproofright<\wd\proofrulename
\then   \shortenproofright=\wd\proofrulename
\fi
\dimen2=\shortenproofleft \advance\dimen2 by\dimen1
\dimen3=\shortenproofright\advance\dimen3 by\dimen4
\ifproofdots
\then
        \dimen6=\shortenproofleft \advance\dimen6 .5\dimen0
        \setbox1=\vbox to\proofdotseparation{\vss\hbox{$\cdot$}\vss}%
        \setbox0=\hbox{%
                \advance\dimen6-.5\wd1
                \kern\dimen6
                $\vcenter to\proofdotnumber\proofdotseparation
                        {\leaders\box1\vfill}$%
                \unhbox\proofrulename}%
\else   \dimen6=\fontdimen22\the\textfont2 
        \dimen7=\dimen6
        \advance\dimen6by.5\proofrulebreadth
        \advance\dimen7by-.5\proofrulebreadth
        \setbox0=\hbox{%
                \kern\shortenproofleft
                \ifdoubleproof
                \then   \hbox to\dimen0{%
                        $\mathsurround0pt\mathord=\mkern-6mu%
                        \cleaders\hbox{$\mkern-2mu=\mkern-2mu$}\hfill
                        \mkern-6mu\mathord=$}%
                \else   \vrule height\dimen6 depth-\dimen7 width\dimen0
                \fi
                \unhbox\proofrulename}%
        \ht0=\dimen6 \dp0=-\dimen7
\fi
\let\doll\relax
\ifwasinsideprooftree
\then   \let\VBOX\vbox
\else   \ifmmode\else$\let\doll=$\fi
        \let\VBOX\vcenter
\fi
\VBOX   {\baselineskip\proofrulebaseline \lineskip.2ex
        \expandafter\lineskiplimit\ifproofdots0ex\else-0.6ex\fi
        \hbox   spread\dimen5   {\hfi\unhbox\proofabove\hfi}%
        \hbox{\box0}%
        \hbox   {\kern\dimen2 \box\proofbelow}}\doll%
\global\dimen2=\dimen2
\global\dimen3=\dimen3
\egroup 
\ifonleftofproofrule
\then   \shortenproofleft=\dimen2
\fi
\shortenproofright=\dimen3
\onleftofproofrulefalse
\ifinsideprooftree
\then   \hskip.5em plus 1fil \penalty2
\fi
}


%% file: local.tex
% ...
\usepackage{mathrsfs}

\mathcode`\<="4268  
\mathcode`\>="5269  

\mathchardef\gt="313E
\mathchardef\lt="313C

\newcommand{\NN}{\mathbb{N}}
\newcommand{\CC}{\mathbb{C}}
\newcommand{\RR}{\mathbb{R}}

\newcommand{\zrule}[2]{%
  #1 \ \  #2}
\newcommand{\urule}[3]{%
  \prooftree #1 \justifies #2 \using #3 \endprooftree}
\newcommand{\brule}[4]{%
  \prooftree #1\ \ \ #2 \justifies #3 \using #4 \endprooftree}
\newcommand{\trule}[5]{%
  \prooftree #1\ \ \ #2 \ \ #3\justifies #4 \using #5 \endprooftree}

\newcommand{\vseq}[1]{%
\begin{array}{c} #1\end{array}} %

\newcommand{\tupla}[1]{\mathsf{t}_{#1}}
\newcommand{\rapp}{\mathsf{r.a}}
\newcommand{\lapp}{\mathsf{l.a}}
\newcommand{\innew}{\mathsf{in.new}}
\newcommand{\inlambda}{\mathsf{in.\lambda}}
\newcommand{\Uq}{\mathsf{Uq}}
\newcommand{\nw}{\mathsf{new}}

\newcommand{\lbeta}{\mathsf{l.\beta}}
\newcommand{\qbeta}{\mathsf{q.\beta}}
\newcommand{\cbeta}{\mathsf{c.\beta}}
\newcommand{\lcom}{\mathsf{l.cm}}
\newcommand{\rcom}{\mathsf{r.cm}}

\newcommand{\Ok}[1]{#1 \in \conf}
\newcommand{\Qr}{{\cal Q}}
\newcommand{\Rr}{{\cal R}}

\newcommand{\Sr}{{\cal S}}

\newcommand{\QV}{{\cal QV}}
\newcommand{\RV}{{\cal RV}}
\newcommand{\SV}{{\cal SV}}

\newcommand{\HS}{\mathcal{H}}
\newcommand{\Hs}[1]{\mathcal{H}(#1)}

\newcommand{\unoqr}{\mathsf{1}}

\newcommand{\Qvt}[1]{\mathbf{Q}(#1)}

\newcommand{\NF}{\mathsf{NF}}
\newcommand{\dr}[1]{|#1\rangle}

\newcommand{\redrul}{\mathscr{L}}

\newcommand{\commrul}{\mathscr{K}}
\newcommand{\noncommrul}{\mathscr{N}}

\newcommand{\wfj}[1]{#1 \triangleright\,}
\newcommand{\qcalc}{\textsf{Q}}

\newcommand{\mixone}{\mathscr{M}}
\newcommand{\mixtwo}{\mathscr{M}'}

\newcommand{\confone}{C}
\newcommand{\conftwo}{D}
\newcommand{\confthree}{E}
\newcommand{\conffour}{F}
\newcommand{\conffive}{G}
\newcommand{\confsix}{H}
\newcommand{\pcompone}{P}
\newcommand{\pcomptwo}{R}
\newcommand{\pcompthree}{Q}
\newcommand{\pcompfour}{S}
\newcommand{\pcompfive}{T}
\newcommand{\pcompsix}{U}
\newcommand{\pcompseven}{V}

\newcommand{\ifte}[3]{\mbox{\texttt{ if }} {#1} \mbox{\texttt{  then} } {#2} \mbox{\texttt{ else} } {#3}}
\newcommand{\meas}[1]{\mathtt{meas}(#1)}

\newcommand{\ifl}{\mathsf{if_1}}

\newcommand{\ifr}{\mathsf{if_0}}

\newcommand{\nonmeasrul}{n\mathscr{M}}

\newcommand{\measr}{\mathsf{meas}}

\newcommand{\mis}[2]{\mathcal{M}_{{#1,#2}}}
\newcommand{\mismin}[2]{\mathsf{m}_{{#1,#2}}}

\newcommand{\conf}{\mathbf{Conf}}

\newcommand{\Mix}{\mathbf{Mix}}

\newcommand{\probmc}[2]{\mathcal{P}(#1,#2)}
\newcommand{\probmcany}[1]{\mathcal{P}(#1)}
\newcommand{\probmcf}{\mathcal{P}}
\newcommand{\numlv}[2]{\mathcal{N}(#1,#2)}
\newcommand{\numlvany}[1]{\mathcal{N}(#1)}
\newcommand{\numlvf}{\mathcal{N}}

\newcommand{\weight}[1]{\mathsf{W}(#1)}
\newcommand{\bdegree}[1]{\mathsf{B}(#1)}

\newenvironment{varitemize}
{
\begin{list}{\condinc{\labelitemi}{\labelitemi}}
{\setlength{\itemsep}{0pt}
 \setlength{\topsep}{0pt}
 \setlength{\parsep}{0pt}
 \setlength{\partopsep}{0pt}
 \setlength{\leftmargin}{15pt}
 \setlength{\rightmargin}{0pt}
 \setlength{\itemindent}{0pt}
 \setlength{\labelsep}{5pt}
 \setlength{\labelwidth}{10pt}
}}
{
 \end{list} 
}

\newcounter{numberone}

\newenvironment{varenumerate}
{
\begin{list}{\arabic{numberone}.}
{
  \usecounter{numberone}
  \setlength{\itemsep}{0pt}
  \setlength{\topsep}{0pt}
  \setlength{\parsep}{0pt}
  \setlength{\partopsep}{0pt}
  \setlength{\leftmargin}{15pt}
  \setlength{\rightmargin}{0pt}
  \setlength{\itemindent}{0pt}
  \setlength{\labelsep}{5pt}
  \setlength{\labelwidth}{15pt}
}}
{
\end{list} 
}

\newcounter{numbertwo}

\newcommand{\new}[1]{\mathtt{new}(#1)}

\newcommand{\redto}{\to}
\newcommand{\ppredto}[2]{\to_{#1}^{#2}}
\newcommand{\predto}[1]{\to_{#1}}

\newcommand{\parred}{\Longmapsto}

\condinc{
\newtheorem{theorem}{Theorem}
\newtheorem{proposition}{Proposition}
\newtheorem{definition}{Definition}
\newtheorem{remark}{Remark}
\newenvironment{proof}{\begin{trivlist}
       \item[\hskip \labelsep {\bfseries Proof.}]}{\hfill $\Box$ \end{trivlist}}}

\newdimen\netunit
\newlength\vt\newlength\hz
\newcommand{\lspine}[3]{%
  \save\POS
  @+, s0+<1em,1.5em>@+, s0+/u#3\vt/+/d1.5em/@+, s0+/l#2\hz/@+, 
  s0+/d#3\vt/@+, s0.{s0+/r#2\hz/}+R(.4)@+, s0+/u1.5em/@+,
  s6+<-1em,1.5em>@+,
  {s0 \ar@{-} 's1 's2 's3 's4 's5 's6 's7 s0},
  {s6.s4+C}.s1.s5="SPINE#1",
  s2.s3+C="SPINE#1-AUX",
  @-@-@-@-@-@-@-
  \restore\POS
  }

\newcommand{\rspine}[3]{%
  \save\POS
  @+, s0+<-1em,1.5em>@+, s0+/u#3\vt/+/d1.5em/@+, s0+/r#2\hz/@+, 
  s0+/d#3\vt/@+, s0.{s0+/l#2\hz/}+L(.4)@+, s0+/u1.5em/@+,
  s6+<1em,1.5em>@+,
  {s0 \ar@{-} 's1 's2 's3 's4 's5 's6 's7 s0},
  {s6.s4+C}.s1.s5="SPINE#1",
  s2.s3+C="SPINE#1-AUX",
  @-@-@-@-@-@-@-
  \restore\POS
  }

